\documentclass[aps,showpacs,prl,twocolumn]{revtex4-1}

\usepackage{exscale}
\usepackage{graphicx}
\usepackage{xcolor}
\usepackage{amsmath}

\usepackage{hyperref}
\hypersetup{
    colorlinks=true,
    linkcolor=black,
    filecolor=black,      
    urlcolor=black,
    citecolor=black,
}

\usepackage{latexsym}
\usepackage{epstopdf}
\usepackage{amsfonts}
\usepackage{color}
\usepackage{amssymb}
\usepackage{bbm}
\usepackage{times}
\usepackage[T1]{fontenc}
\usepackage{lipsum}
\usepackage{amsthm}
\usepackage{fancyhdr,txfonts,bbm}
\usepackage{appendix}
\usepackage{changes}
\definechangesauthor[name={Markus}, color=red]{M}
\definechangesauthor[name={Dario}, color=blue]{D}

\newcommand{\lrho}{\rho_{AB}^{(\lambda)}}

\newcommand{\PP}{\mathcal{P}}
\newcommand{\oPP}{{\overline{\mathcal{P}}}}
\newcommand{\kbra}[1]{| #1\rangle\!\langle #1|}

\newcommand{\PAi}{\{ P_{A,i} \}_i}

\newcommand{\bra}[1]{\langle #1|}
\newcommand{\ket}[1]{|#1\rangle}

\newcommand{\ketbra}[2]{| #1\rangle\!\langle #2|}
\newcommand{\tr}[1]{\mbox{Tr}\left[ #1\right]}

\newcommand{\trA}[1]{\mbox{Tr}_A\left[ #1\right]}

\newcommand{\trB}[1]{\mbox{Tr}_B\left[ #1\right]}

\newtheorem{theorem}{Theorem}

\newtheorem{definition}{Definition}

\begin{document}

\title{Equivalence between non-Markovian dynamics and correlation backflows}
\author{Dario De Santis$^1$ and Markus Johansson$^1$}
\affiliation{$^1$ICFO-Institut de Ciencies Fotoniques, The Barcelona Institute of Science and Technology,
08860 Castelldefels (Barcelona), Spain}
\date{\today}

\begin{abstract}

The information encoded into an open quantum system that evolves under a Markovian dynamics is always monotonically non-increasing.
Nonetheless, for a given quantifier of the information contained in the system, it is in general not clear if for all non-Markovian dynamics it is possible to observe a non-monotonic evolution of this quantity, namely a backflow.
We address this problem by considering correlations of finite-dimensional bipartite systems.
For this purpose, we consider a class of correlation measures and prove that if the dynamics is non-Markovian there exists at least one element from this class that provides a correlation backflow.
Moreover, we provide a set of initial probe states that accomplish this witnessing task.
This result provides the first one-to-one relation between non-Markovian dynamics of finite-dimensional quantum systems and correlation backflows.

\end{abstract}

\maketitle

The study of open quantum systems dynamics \cite{BPbook,RHbook} is of central interest in quantum mechanics. A quantum system is called open when interaction with the environment that surrounds the quantum system is included in the description of its evolution. Since there are no experimental scenarios where a quantum system can be considered completely isolated, this approach provides a more realistic description of quantum evolutions. 

The interaction between an open quantum system $S$ and its environment $E$ leads to two possible regimes of evolution. The phenomena associated with the \textit{Markovian} regime are characterized by the monotonic non-increase of the information contained in the open system. In this case we have a unidirectional flow of information away from $S$ and we say that the dynamics is memoryless. Instead, in the \textit{non-Markovian} regime, this flow is not unidirectional and part of the information lost is recovered in one or more subsequent time intervals. This phenomenon is called \textit{backflow} of information.
However, it is nonobvious what mathematical framework is better suited to reproduce this phenomenology. Recently, a framework based on a notion of \emph{divisibility} of dynamical maps, namely the operators describing the dynamical evolution of the system, achieved a promising consensus \cite{RHbook,WWTA,INI,RHP,BLP,BD,bogna,LFS,BognaREV}. More precisely, it requires that, if the dynamics is Markovian, the evolution between any two times is represented by a completely positive and trace-preserving (CPTP) linear map.

Many efforts are directed towards testing this mathematical definition by studying the characteristic backflows of information that different physical quantities show when the evolution is non-Markovian. Once we consider a quantity that is non-increasing under Markovian evolutions, we can study its ``non-Markovian witnessing potential'', namely the ability to show a backflow when the dynamics is non-Markovian. 
Distinguishability between states \cite{BLP,BD,bogna}, correlation measures \cite{LFS,PRA,long,Janek}, channel capacities \cite{BognaChannel0} and the volume of accessible states \cite{Volume} are some examples of quantities that have been studied in this scenario. Moreover, while Markovian phenomena are reproduced correctly by definition, the non-trivial point that has to be analyzed is if it is possible to obtain one-to-one connections between backflows of these quantities and non-Markovian dynamical maps. Indeed, this result  would imply a correspondence between the phenomenological and the mathematical description of non-Markovianity that we have presented.

In this work we focus on the witnessing potential of the set of correlation measures. In particular, we study the connection between revivals of bipartite correlations and when the evolution of one subsystem $S$ is non-Markovian. 
 Several measures have already been considered in this scenario, e.g. quantum mutual information \cite{LFS,long} and entanglement measures \cite{Janek}. Recently, a correlation measure  that witnesses almost all non-Markovian dynamics has been introduced \cite{PRA}. However, it is unknown if any of these correlation measures can witness all non-Markovian dynamics \cite{long}.

The main result of this work is the first proof of a one-to-one relation between correlation backflows and non-Markovian dynamics.
We consider a class of correlation measures for bipartite systems that provides backflows if and only if the dynamics is not Markovian. For this purpose, we make use of supplementary ancillary systems to define initial probe states that allow to succeed in this witnessing task. Finally, we introduce a measure of non-Markovianity.

{\textit{Non-Markovianity and divisibility properties.---}}Given a generic finite-dimensional Hilbert space $\mathcal{H}$, we define $\mathcal{B}(\mathcal{H})$ to be the set of linear bounded operators that act on $\mathcal{H}$ and $\mathcal{S}(\mathcal{H})$ the set of positive semidefinite, Hermitian and trace one operators on $\mathcal{H}$, namely the state space of $\mathcal{H}$.

We consider an open quantum system $S$ described by states on a finite-dimensional Hilbert space $\mathcal{H}_S$. At the initial time $t_0$ the system $S$ is uncorrelated with the surrounding environment $E$. The evolution of $S$ from $t_0$ to $t\geq t_0$ is given by a dynamical map: a CPTP linear operator $\Lambda_S(t,t_0) : \mathcal{S}(\mathcal{H}_S)\rightarrow \mathcal{S}(\mathcal{H}_S)$. Therefore, the complete evolution of $S$, namely from $t_0$ to any time $t\geq t_0$, is described by a family of dynamical maps  $\{ \Lambda_S(t,t_0)\}_t$, where $\Lambda_S(t,t_0)$ is CPTP for every $t\geq t_0$.

The concept needed to define the mathematical structure we adopt to define Markovianity  is the completely positive (CP) divisibility of the family $\{\Lambda_S(t,t_0)\}_t$  in terms of intermediate maps $V_S(t,t')$.

\begin{definition}
The evolution $\{\Lambda_S(t,t_0)\}_t$ is called CP-divisible if, for any $t\geq  t_0$, the dynamical map $\Lambda_S(t ,t_0)$ can
be decomposed as a sequence of CPTP linear maps $\Lambda_S(t,t_0)=V_S(t,t') \, \Lambda_S(t',t_0)$, where $V_S(t,t')$ is a CPTP linear map for any $t_0 \leq t'\leq t$. \end{definition}
CP-divisibility is commonly used to define Markovian dynamics and it is the definition that  we consider in this work: $\{\Lambda_S(t,t_0)\}_t$  is Markovian if and only if it is CP-divisible. Likewise, we call an evolution non-Markovian if and only if for some $t_0 \leq t'\leq t$ there is no CPTP intermediate map $V_S(t,t')$.

{\textit{Measurements with fixed output probability distributions.---}}Any measurement process on a quantum state $\rho \in \mathcal S(\mathcal{H})$ is defined by a \textit{positive-operator valued measure} (POVM), namely an indexed set of Hermitian and positive semi-definite operators $\{P_i\}_{i=1}^n$ of $ \mathcal{B}(\mathcal{H})$ such that $\sum_{i=1}^n P_i = \mathbbm{1}$, where $\mathbbm{1} \in \mathcal{B}(\mathcal{H})$ is the identity operator on $\mathcal{H}$ and $n$ is the number of possible  measurement outcomes.   The operator $P_i$ represents the $i$-th output of the measurement, where $p_i=\tr{\rho P_i}$ is the corresponding occurrence probability.

Let $\mathcal{E}=\{p_i,\rho_i\}_{i=1}^n$ be a generic ensemble of $n$ states where each finite-dimensional state $\rho_i\in \mathcal{S}(\mathcal{H})$ occurs with probability $p_i$.
Now we consider a bipartite state  $\rho_{AB} \in \mathcal{S}(\mathcal{H}_A \otimes \mathcal{H}_B)$ and a POVM $\{P_{A,i}\}_{i=1}^n$ defined for the subsystem $A$. We define $\mathcal{E}(\rho_{AB}, \{P_{A,i}\}_{i=1}^n)\equiv \{ p_i , \rho_{B,i}\}_{i=1}^n$ to be the ensemble of states of $B$ that we obtain when we apply on $A$ the measurement $\{P_{A,i}\}_{i=1}^n$, where  
\begin{equation}\label{ens}
p_i= \tr{ \rho_{AB}  P_{A,i} \otimes \mathbbm{1}_B} , \,\,\, \rho_{B,i} = {\trA{ \rho_{AB}  P_{A,i} \otimes \mathbbm{1}_B}}/{ p_i }\, .
\end{equation}
We call  $\{p_i\}_{i=1}^n$ and $\{ \rho_{B,i}\}_{i=1}^n$ respectively the output probability distribution and the output states of the measurement. We call their combination $\mathcal{E}(\rho_{AB} , \{P_{A,i}\}_{i=1}^n)$ the output ensemble.

We consider finite probability distributions $\PP=\{ p_i\}_{i=1}^n$ composed by $n$ positive elements, where $\sum_{i=1}^n p_i=1$. We define the set of $n$-output POVMs that, if applied on $\rho\in \mathcal{S}(\mathcal{H})$, provide $\PP$-distributed outcomes.

\begin{definition}\label{probdistribution}
Given the finite probability distribution $\mathcal{P}=\{p_i\}_{i=1}^n$, the $n$-output POVM $\{P_{i}\}_{i=1}^n$ on $\mathcal{H}$ is a $\PP$-POVM  for  $\rho\in \mathcal{S}(\mathcal{H})$ if and only if it belongs to 
$$
\Pi^\PP(\rho)\equiv \{\{ P_i\}_{i=1}^n : \, \tr{\rho \, P_i} = p_i, \forall i=1,...\, ,n \} \, .
$$
\end{definition}

Similarly, given a bipartite system state $\rho_{AB}$, we define the measurement processes that, if applied on one side of $\rho_{AB}$, provide $\PP$-distributed output ensembles (see Fig. \ref{scenariofig}).
\begin{definition}\label{probdistribution}
Given the finite probability distribution $\mathcal{P}=\{p_i\}_{i=1}^n$, the $n$-output POVM $\{P_{A,i}\}_{i=1}^n$ on $\mathcal{H}_A$   is a $\PP$-POVM on $A$ for $\rho_{AB}\in \mathcal{S}(\mathcal{H}_{AB})$ if and only if it belongs to 
$$
\Pi^\PP_A(\rho_{AB})\equiv \{\{ P_{A,i}\}_{i=1}^n : \, \tr{\rho_{AB} \, P_{A,i} \otimes \mathbbm{1}_B} = p_i , \forall i=1,\dots ,n \} \, .
$$
\end{definition} 
Analogously, we can define $\Pi^\PP_B(\rho_{AB})$. 
We notice that for any given $\PP$ and  $\rho_{AB}$, we have 
$\Pi^\PP_A(\rho_{AB})= \Pi^\PP(\rho_A) \, ,
$ where $\rho_A=\trB{\rho_{AB}}$. Moreover, $\Pi^\PP(\rho)$ ($\Pi_A^\PP(\rho_{AB})$) is a non-empty convex set for any $\rho$ ($\rho_{AB}$) and $\PP$.

\begin{figure}
\includegraphics[width=0.34\textwidth]{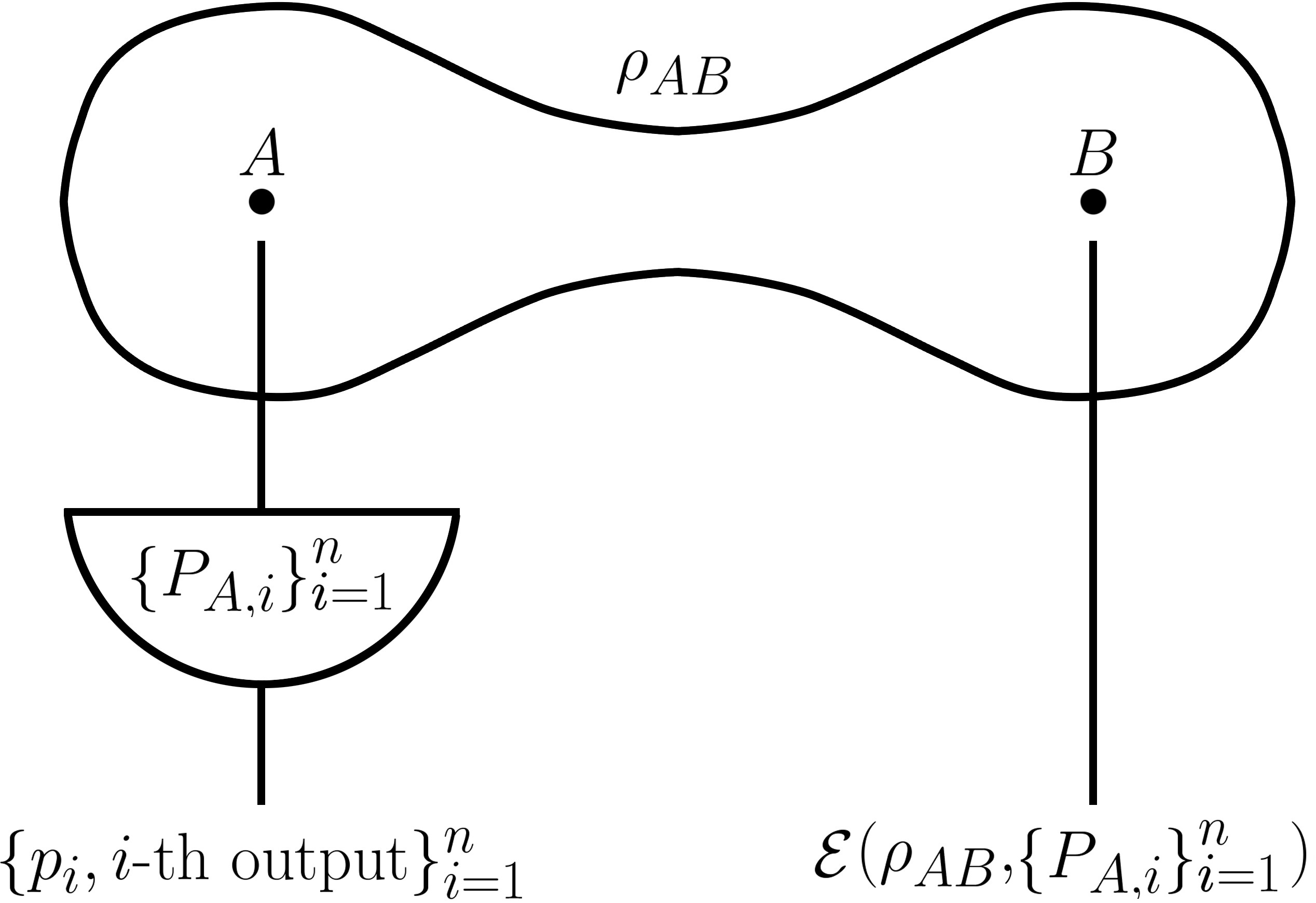}
\caption{Given a probability distribution $\PP=\{p_i\}_{i=1}^n$, $\{P_{A,i}\}_{i=1}^{ n}$ is a $\PP$-POVM for $\rho_{AB}$ if and only if the output probability distribution of this measurement is $\PP$. The correlation $C_A^\PP(\rho_{AB})$ considers the scenario where $\rho_{AB}$ is measured with a $\PP$-POVM $\{P_{A,i}\}_{i=1}^{ n}$ on $A$  that provides the largest guessing probability of the corresponding output ensemble $\mathcal{E}(\rho_{AB},\{P_{A,i}\}_{i=1}^{ n})=\{p_i,\rho_{B,i}\}_{i=1}^{ n}$ of states of $B$.}\label{scenariofig}
\end{figure}

{\textit{Witnessing non-Markovianity with  distinguishability of ensembles.---}}  We apply an $n$-output measurement $\{P_i \}_{i=1}^n$ on a state that we randomly extract from an ensemble $\mathcal{E}=\{p_i, \rho_i\}_{i=1}^{n}$ of states of $\mathcal{S}(\mathcal{H})$. The \textit{guessing probability} $P_g(\mathcal{E})$ is the average probability to successfully identify the extracted state with an optimal measurement. This quantity is defined as
\begin{equation}\label{Pgens}
P_g(\mathcal{E}) \equiv \max_{\{P_i \}_{i=1}^n} \sum_{i=1}^n p_i \, \tr{\rho_i \, P_i} \, ,
\end{equation}
where the maximization is performed over the $n$-output  POVMs of $\mathcal{B}(\mathcal{H})$.

Now we describe how guessing probability can be used to witness non-Markovianity.
We consider a finite-dimensional system $\mathcal{H}_S\otimes\mathcal{H}_{A'}$, where the open quantum system $S$ is evolved by a generic {$\{\Lambda_S(t,t_0)\}_t$} and ${A'}$ is an ancillary system. Given an initial ensemble $\mathcal{E}_{SA'}(t_0)=\{p_i, \rho_{SA',i}\}_i$, we consider its evolution:
\begin{equation}\label{ensevo}
\mathcal{E}_{SA'}(t_0) \, \longrightarrow \, \mathcal{E}_{SA'}(t)= \{ p_i , \Lambda_S(t,t_0) \otimes {I}_{A'}( \rho_{SA',i} )\}_i \,  ,
\end{equation}
where $I_{A'}:\mathcal{S}(\mathcal{H}_{A'}) \rightarrow \mathcal{S}(\mathcal{H}_{A'})$ is the identity map on $\mathcal{S}(\mathcal{H}_{A'})$. For any CPTP map $\Lambda$ acting on the states of $\mathcal{E}=\{p_i,\rho_i\}_i$, $P_g(\mathcal{E})$ is non-increasing: $P_g(\{p_i,\rho_i\}_i)\geq P_g(\{p_i,\Lambda(\rho_i)\}_i)$. Therefore,  if $\{\Lambda_S(t,t_0)\}_t$  is CP-divisible,
\begin{equation}
 P_g (\mathcal{E}_{SA'} (\tau+\Delta \tau)) - P_g (\mathcal{E}_{SA'} (\tau))  \leq 0 \, ,
\end{equation}
for every $\tau\geq t_0$ and $\Delta \tau\geq 0$.

 Given any {evolution $\{\Lambda_S(t,t_0)\}_t$} and time interval $[\tau, \tau+\Delta \tau]$, there exist an ancillary system $A'$ and an initial ensemble $\overline{\mathcal{E}}_{S A'} (t_0)$ of separable states of $\mathcal{S}(\mathcal{H}_S\otimes\mathcal{H}_{A'})$
\begin{equation}\label{ensBD}
\overline{\mathcal{E}}_{S A'} (t_0) \equiv \{\overline{p}_i, \overline{\rho}_{SA',i} \}_{i=1}^{\overline{n}}  \, ,
\end{equation}
such that we have a backflow 
\begin{equation}\label{PgBD}
 P_g(\overline{\mathcal{E}}_{S A'}(\tau + \Delta \tau)) -P_g(\overline{\mathcal{E}}_{SA'}(\tau))  > 0 \, ,
\end{equation}
if and only if there is no CPTP intermediate map $V_S(\tau+\Delta\tau,\tau)$, as shown in Ref. \cite{BD}. Moreover, $\oPP \equiv \{\overline{p}_i\}_{i=1}^{ \overline{n}}$ is finite and $\mbox{dim}(\mathcal{H}_{A'})\leq d_S \equiv \mbox{dim}(\mathcal{H}_S)$. We underline that, even if  we do not make it explicit, $\overline{\mathcal{E}}_{S A'} (t_0)$ strictly depends on {$\{\Lambda_S(t,t_0)\}_t$} and $[\tau,\tau+\Delta \tau]$.
The result of Ref. \cite{BD} is general and applies to any evolution defined on a finite-dimensional system.

{\textit{A class of correlation measures.---}}
Let $\PP \equiv \{p_i\}_{i}$ be a generic finite probability distribution and $\rho_{AB} \in  \mathcal{S}(\mathcal{H}_A\otimes \mathcal{H}_B)$ a generic finite-dimensional bipartite system state. We consider the correlation measure
\begin{equation}\label{measure2}
C_A^\PP(\rho_{AB}) \equiv \max_{\{P_{A,i} \}_{i} \in \Pi^\PP_A (\rho_{AB}) } P_g  \left(\, \rho_{AB} ,  \{P_{A,i}\}_{i} \right) - p_{max}\, ,
\end{equation}
where the maximization is performed over the $\PP$-POVMs on $A$ for $\rho_{AB}$ and we used the definitions $P_g(\rho_{AB}, \{P_{A,i}\}_{i} )\equiv P_g(\mathcal{E}(\rho_{AB} ,  \{P_{A,i}\}_{i}))$ and $p_{max}\equiv \max_i p_i$ (see Fig \ref{scenariofig}). Therefore, we can consider a class of correlation measures where each element is defined by a different distribution $\PP$.

The operational meaning of this correlation measure for a given $\PP$ is the following. Its value (modulo $p_{max}$) is the largest guessing probability of the ensembles $\{p_i,\rho_{B,i}\}_i$  on $B$ that $A$ can generate measuring its side of $\rho_{AB}$ with $\PP$-POVMs. Therefore, $C_A^\PP(\rho^{(1)}_{AB}) > C_A^\PP(\rho^{(2)}_{AB}) $ implies that  the largest distinguishability of the $\PP$-distributed output ensembles of $B$ that $A$ can generate measuring $\rho^{(1)}_{AB}$ is greater than the largest distinguishability of the $\PP$-distributed output ensembles of $B$ that $A$ can generate measuring $\rho^{(2)}_{AB}$.

To consider $C^\PP_A$  a proper correlation measure, we have to show that it is: zero-valued for product states, non-negative and monotonically decreasing under local operations \cite{long}. In order to prove the first property, given a generic product state $\rho_{AB}=\rho_A\otimes \rho_B$, the output ensemble $\mathcal{E}(\rho_A\otimes \rho_B , \{P_{A,i}\}_i ) = \{p_i, \rho_B\}_i$ is made of identical states  for any POVM $\{P_{A,i}\}_i$ and $P_g (\{p_i, \rho_B\}_i)=p_{max}$.  Therefore, while  $C_A^\PP(\rho_{AB})\geq 0$ is now trivial, the proof for the monotonicity of $C_A^\PP(\rho_{AB})$ under local operations is in the Supplemental Material (SM).

Similarly, we can define the class of measures of the form
\begin{equation}\label{measureB}
C_B^\PP(\rho_{AB}) \equiv \max_{\{P_{B,i} \}_i \in \Pi^\PP_B (\rho_{AB}) } P_g  \left(\, \rho_{AB} ,  \{P_{B,i}\}_{i=1}^n \right) - p_{max}\, .
\end{equation}
Since in general $C_A^\PP(\rho_{AB})\neq C_B^\PP(\rho_{AB})$, we can consider the symmetric class of measures
\begin{equation}\label{measuresymm}
C_{AB}^\PP(\rho_{AB}) \equiv \max\left\{ C_A^\PP(\rho_{AB}),C_B^\PP(\rho_{AB}) \right\} \, .
\end{equation}
Finally, we notice that the correlation measures given in Eqs. (\ref{measure2}), (\ref{measureB}) and (\ref{measuresymm})
can be considered as generalizations for generic distributions $\PP$ of the correlation measures introduced in Ref. \cite{PRA}, where only uniform distributions are considered.

{\textit{The probe states.---}}The goal of this work is to prove a one-to-one correspondence between non-Markovianity and correlation backflows. Therefore, similarly to Ref. \cite{BD}, we consider the most general scenario where a family of dynamical maps $\Lambda_S(t,t_0)$ defines the evolution for $t\geq t_0$ and we focus on a generic time interval $[\tau, \tau+\Delta \tau]$. We provide an \textit{initial probe state} and a distribution $\PP$ for which the correlation measure $C^\PP_A$ shows a backflow in the time interval $[\tau,\tau+\Delta \tau]$ if and only if there is no CPTP intermediate map $V_S(\tau+\Delta \tau, \tau)$.

First, we introduce the bipartition and the state space needed to consider $C^\PP_A$ and the initial probe state. 
We define the bipartite system $\mathcal{S}(\mathcal{H}_A \otimes~ \mathcal{H}_B)$ such that  dim$(\mathcal{H}_A)=~\overline{n} \, $
 and $\mathcal{H}_B \equiv   \mathcal{H}_S \otimes \mathcal{H}_{A'} \otimes\mathcal{H}_{A''}$, where $\mbox{dim}(\mathcal{H}_{S})=\mbox{dim}(\mathcal{H}_{A'})=
d_S$ and  dim$(\mathcal{H}_{A''})=\overline{n}+1  $.  
 We fix the following orthonormal basis for $\mathcal{H}_{A}$ and  $\mathcal{H}_{A''}$: $\mathcal{M}_{A} \equiv \{\ket{i}_{A} \}_{i=1}^{\overline{n}} = \{\ket{1}_A, \ket{2}_A, ...\, ,\ket{\overline{n}}_A\}$  and 
$\mathcal{M}_{A''}\equiv \{\ket{i}_{A''} \}_{i=1}^{\overline{n} +1 } = \{\ket{1}_{A''}, \ket{2}_{A''}, ...\, ,\ket{\overline{n}+1}_{A''}\}$. Notice that the ancillas $A'$ and $A''$ can be considered as a single ancilla with Hilbert space $\mathcal{H}_{A'}\otimes \mathcal{H}_{A''}$ (see Fig. \ref{fig1}).

\begin{figure}
\includegraphics[width=0.3\textwidth]{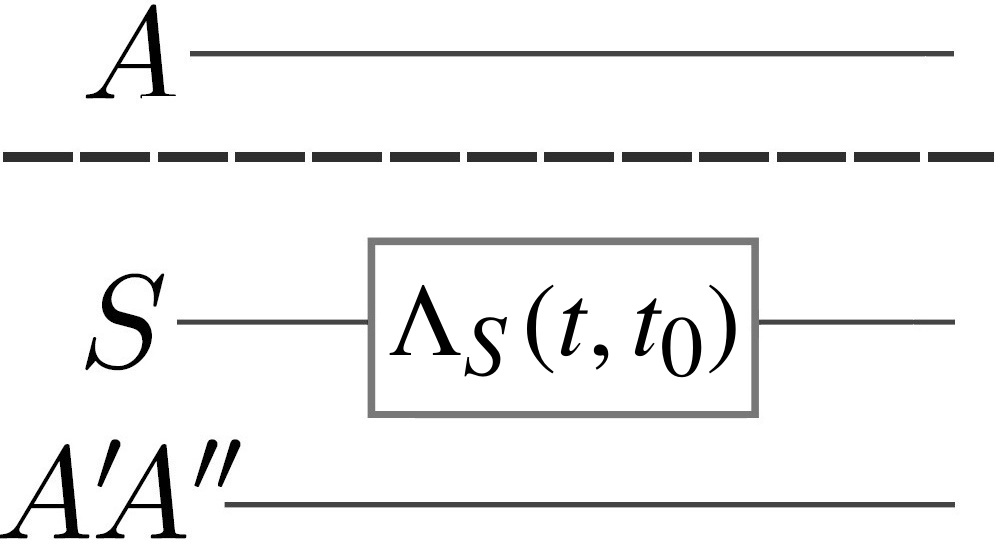}
\caption{The initial probe state $\rho_{AB}^{(\lambda)}(t_0)$ belongs to the bipartite system $ \mathcal{S}(\mathcal{H}_A\otimes \mathcal{H}_B)$,  where $\mathcal{H}_B=\mathcal{H}_S\otimes\mathcal{H}_{A'} \otimes\mathcal{H}_{A''}$. We consider the correlation $C_A^\oPP(\rho_{AB}^{(\lambda)}(t))$ given by the bipartition between the subsystems $A$ and $B$, where the open quantum system  $S$ undergoes the evolution defined by $\Lambda_S(t,t_0)$.}\label{fig1}
\end{figure}

We define $\overline{\rho}_{B,i}  \equiv    \overline{\rho}_{SA',i}\otimes \ketbra{\overline{n}+1}{\overline{n}+1}_{A''} \in \mathcal{S}(\mathcal{H}_B)$, for $i=1,\dots,\overline n$, where we made use of the elements of  $\overline{\mathcal{E}}_{S A'} (t_0) = \{\overline{p}_i, \overline{\rho}_{SA',i} \}_{i=1}^{\overline{n}}$  (see Eq. (\ref{ensBD})). We introduce a  class of initial probe states $\rho_{AB}^{(\lambda)} (t_0) \in \mathcal{S}(\mathcal{H}_A \otimes \mathcal{H}_B)$ parametrized by  $\lambda \in [0,1)$
\begin{equation}\label{probebi}
\rho_{AB}^{(\lambda)} (t_0) \equiv \sum_{i=1}^{\overline{n}} \overline{p}_i \, \ketbra{i}{i}_A \otimes  \left(\lambda\, \sigma_{SA'}  \otimes \ketbra{i}{i}_{A''} + (1-\lambda)\,  \overline{\rho}_{B,i}  \right) ,
\end{equation}
where $\sigma_{SA'}$ is a generic state of $\mathcal{S}(\mathcal{H}_S \otimes \mathcal{H}_{A'})$. Notice that in Eq. (\ref{probebi}) the index $i$  runs from $1$ to $\overline{n}$.  Since the ancillary systems do not evolve, the action of the dynamical map of the evolution on the probe state, i.e., $I_{ A}\otimes \Lambda_S(t,t_0)\otimes I_{ A'A''} (\rho_{AB}^{(\lambda)} (t_0)) $,  preserves the initial classical-quantum separable structure for any $t\geq t_0$
\begin{equation}\label{probebit}
\rho_{AB}^{(\lambda)} (t) = \sum_{i=1}^{\overline{n}} \overline{p}_i \, \ketbra{i}{i}_A \otimes \left( \lambda\, \sigma_{SA'}(t)\otimes \ketbra{i}{i}_{A''} + (1-\lambda) \, \overline{\rho}_{B,i} (t) \right) ,
\end{equation}
where $\overline{\rho}_{B,i} (t) =\Lambda_S(t,t_0)\otimes I_{ A' A''} \,( \overline{\rho}_{B,i}  )$ and $\sigma_{SA'}(t)=\Lambda_S(t,t_0)\otimes I_{ A' } (\sigma_{SA'})$. Finally, since $\mbox{Tr}_B [ \lrho (t) ]=\sum_{i=1}^{\overline{n}} \overline p_i \ketbra{i}{i}_A$, the set $\Pi^\PP_A ( \lrho (t) ) = \Pi^\PP (\mbox{Tr}_B [ \lrho (t) ] )$ does not depend on $t$ and $\lambda$. 

{\textit{Witnessing non-Markovianity with correlations.---}}We provide a procedure that witnesses any non-Markovian dynamics with a correlation backflow. 
In the case of bijective or pointwise non-bijective $\Lambda_S(t,t_0)$,
 this scenario has been studied in  Refs. \cite{PRA,Janek}. Moreover, the  negativity entanglement measure witnesses any non-Markovian qubit evolution \cite{Janek}.

In order to witness non-Markovianity through backflows of $C_A^\oPP$, the evolution of the initial state $\rho_{ASA'}=\sum_{i=1}^{\overline n} \overline p_i \ketbra{i}{i}_A \otimes \overline\rho_{SA',i} $ is an intuitive choice. Indeed, $\{\ketbra{i}{i}_{A} \}_{i=1}^{\overline n} \in \Pi_A^\oPP(\rho_{ASA'}(t))$ for all $t\geq t_0$ and $P_g(\rho_{ASA'}(t) ,\{\ketbra{i}{i}_{A} \}_{i=1}^{\overline n} )= P_g(\overline{\mathcal E}_{SA'}(t))$  (see Eq. (\ref{PgBD})). Nonetheless, in general $\{\ketbra{i}{i}_{A}\}_{i=1}^{\overline n}$ is not selected by the maximization that defines $C^\oPP_A(\rho_{ASA'}(t))$ \cite{long}.

We present the main result of this work, namely that the class of correlation measures $C^\oPP_A$ is able to witness any non-Markovian dynamics. 
\begin{theorem}\label{theorem}
For any {evolution $\{\Lambda_S(t,t_0)\}_t$} defined on a finite-dimensional system $S$ and time interval $[\tau, \tau+\Delta \tau]$ there exist at least one ancillary system $\mathcal{H}$, one bipartite system $\mathcal{H}_A\otimes \mathcal{H}_B$, where  $\mathcal{H}_B=\mathcal{H}_S \otimes \mathcal{H}$, a correlation measure for bipartite systems $\mathcal{C}_{AB}$ and an initial state  $\rho_{AB}(t_0)\in \mathcal{S} (\mathcal{H}_A\otimes \mathcal{H}_B)$ such that a backflow
$$
 \mathcal{C}_{AB} \left(  \rho_{AB} ( \tau+ \Delta \tau ) \right) - \mathcal C_{AB} \left(  \rho_{AB} ( \tau ) \right) > 0 \, ,
$$
occurs if and only if there is no CPTP intermediate map $V_S(\tau+\Delta \tau,\tau)$, where $S$ is the only system that evolves during the evolution.
\end{theorem}
\begin{proof}
We consider the ancillary system $\mathcal{H}=\mathcal{H}_{A'}\otimes\mathcal{H}_{A''}$,  the correlation measure $\mathcal{C}_{AB}=C^{\oPP}_{A}$  and the set of  initial probe states $\rho_{AB}^{(\lambda)}(t_0)$.
We prove that, for wisely chosen values of $\lambda$, we have a backflow
\begin{equation}\label{backflowC}
 \Delta {C}_{A}^\oPP \equiv C_{A}^{\oPP} \left( \lrho (\tau+\Delta \tau)\right) - C_{A}^{\oPP} \left(\lrho (\tau) \right) > 0 \, ,
 \end{equation}
if and only if there is no CPTP  intermediate map  $V_S(\tau + \Delta \tau, \tau)$. 

We notice that  $\{ \ketbra{i}{i}_A\}_{i=1}^{\overline{n}}  \in \Pi_A^\oPP( \lrho (t) )$  is a $\oPP$-POVM on $A$ for the probe state. Moreover, as noticed above, $\Pi^\oPP_A( \lrho (t) )$ does not depend on $\lambda$ and $t$. In the following, if not specified otherwise, the index $i$  runs from 1 to $\overline n$.  The output ensemble  that we obtain measuring  $\rho_{AB}^{(\lambda)}(t)$ with $\{ \ketbra{i}{i}_A\}_{i}$ is
\begin{equation}\label{EnsPi}
\mathcal{E} \left(\lrho(t) , \{ \ketbra{i}{i}_A\}_i \right)= \left\{\overline{p}_i \, , \, \lambda \sigma_{SA'}(t) \otimes \ketbra{i}{i}_{A''} + (1-\lambda) \overline{\rho}_{B,i} (t) \right\}_i   .
\end{equation}
The corresponding guessing probability is (See SM)
\begin{equation}\label{PgApix}
P_g \left(\lrho(t) , \{ \ketbra{i}{i}_A\}_i  \right) = \lambda + (1-\lambda) \, P_g\left(\overline{\mathcal{E}}_{SA'}(t) \right)  .
\end{equation}

Now we consider $\{P_{A,i}\}_i \in \Pi_A^\oPP( \lrho (t) )$  different from $\{ \ketbra{i}{i}_A\}_i$. In general, we obtain (See SM):
\begin{equation}\label{EnsJx}
\mathcal{E} \left(\lrho (t), \{P_{A,i} \}_i \right)= \left\{ \overline{p}_i, \lambda \sigma_{B,i}^{\perp }(t) + (1-\lambda) \sigma_{B,i}^{\parallel } (t) \right\}_i  .
\end{equation}
Each state $\sigma_{B,i}^{\perp }(t)$ is defined as $\sigma_{B,i}^{\perp }(t)\equiv \sigma_{SA'} (t) \otimes \rho_{A'',i}^{\perp}$, where  $\rho_{A'',i}^{\perp}$ is a  convex combination of the states $\{ \ketbra{k}{k}_{A''} \}_{k=1}^{\overline n } $. Analogously,  $\sigma_{B,i}^{\parallel }(t) \equiv \rho_{SA',i}^{\parallel} (t) \otimes \ketbra{\overline n +1}{\overline n +1}_{A''}$, where $\rho_{SA',i}^{\parallel} (t)$ is a convex combination of the states $\{\overline{\rho}_{SA',k}(t) \}_{k=1}^{\overline n }$  (See SM). 
Similarly to Eq. (\ref{PgApix}), we obtain
\begin{equation}\label{generalPg}
P_g \left(\lrho(t) , \{P_{A,i} \}_i \right) \! = \!  \lambda P_g  \left( \{ \overline{p}_i , \rho_{A'',i}^{\perp }  \}_i \right) + \!(1-\lambda)  P_g \left( \{ \overline{p}_i , \rho_{SA',i}^{\parallel}(t)  \}_i \right) .
\end{equation}

In order to understand when $C^{\oPP}_A(\lrho ( t ))$ shows a backflow in $[\tau, \tau + \Delta \tau]$, we write:
\begin{equation}\label{deltaC}
 \Delta {C}_{A}^\oPP \geq P_g \left(\lrho(\tau+\Delta \tau) ,\{ \ketbra{i}{i}_A\}_i \right) - P_g \left(\lrho(\tau) , \{ P_{A,i}^{(\lambda)} \}_i \right)\, .
\end{equation}
We focus on $C^{\oPP}_A(\lrho ( t ))$ at $t=\tau$ for different values of $\lambda$ (we omit the dependence on $\tau$ of some quantities to increase readability). We define the ``optimal'' $\oPP$-POVMs $\{ P_{A,i}^{(\lambda)} \}_i $ to be the $\oPP$-POVMs that at $t=\tau$ solve the maximization that defines $C^{\oPP}_A (\lrho ( \tau ) )$
\begin{equation}
C^{\oPP}_A \left(\lrho ( \tau ) \right)=P_g \left(\lrho ( \tau ), \{ P_{A,i}^{(\lambda)} \}_i \right) - \overline{p}_{max} \, .
\end{equation}
We consider Eq. (\ref{generalPg}) when an optimal $\{ P_{A,i}^{(\lambda)} \}_i$ is chosen. We define the corresponding ensembles that appear in this expression { {$\mathcal{E}^{\perp} (\{ P_{A,i}^{(\lambda)} \}_i)$} and { $\mathcal{E}^{\parallel}{(\{ P_{A,i}^{(\lambda)} \}_i)}$}, namely}
\begin{equation}\label{optimalPg}
{P_g \left(\lrho ( \tau ), \{ P_{A,i}^{(\lambda)} \}_i \right) = \lambda P_g (\mathcal{E}^{\perp} (\{ P_{A,i}^{(\lambda)} \}_i)) + (1-\lambda) P_g(\mathcal{E}^{\parallel} (\{ P_{A,i}^{(\lambda)} \}_i) ) \, .}
\end{equation}

We focus on Eq. (\ref{deltaC}) and we distinguish the two possible scenarios:
\begin{itemize}
\item  (A): one of the optimal measurements is $ \{ P_{A,i}^{(\lambda)} \}_i= \{ \ketbra{i}{i}_A \}_i$ for some $\lambda\in[0,1)$,
\item (B): none of the optimal measurements $\{ P_{A,i}^{(\lambda)} \}_i$ is equal to $\{ \ketbra{i}{i}_A \}_i$ for any $\lambda\in[0,1)$.
\end{itemize}
We start studying case (A).
 In SM we prove that if $ \{ \ketbra{i}{i}_A \}_i $ is an optimal $\oPP$-POVM for some $\lambda^*$,  then the same is true for any  $\lambda \in( \lambda^*,1)$. From Eqs. (\ref{PgBD}), (\ref{PgApix}) and (\ref{deltaC}), for $\lambda \in( \lambda^*,1)$
\begin{equation}
\Delta {C}_{A}^\oPP  \geq (1-\lambda) \left( P_g(\overline{\mathcal{E}}_{SA'}(\tau+\Delta \tau)) - P_g(\overline{\mathcal{E}}_{SA'}(\tau)) \right) >0 \, ,
\end{equation}
if and only if there is no CPTP intermediate map $V_S(\tau+\Delta \tau,\tau)$ for $\Lambda_S(t,t_0)$.

Now we analyze case (B).
In SM we show that for $\lambda=1$ the unique optimal $\oPP$-POVM is $ \{ \ketbra{i}{i}_A \}_i $. Moreover, $P_g( \lrho ( \tau ), \{ P_{A,i} \}_i )$ is Lipschitz continuous in $\lambda$ and $P_g( \rho_{AB} , \{ P_{A,i} \}_i )$ is Lipschitz continuous in $\{ P_{A,i} \}_i$. This implies that the set of optimal $\oPP$-POVMs $\{ P_{A,i}^{(\lambda)} \}_i$ is contained in a neighbourhood of $\{ \ketbra{i}{i}_A \}_i $ with size decreasing towards zero as $\lambda$ approaches $1$. 
This in turn implies that the set of guessing probabilities $P_g(\mathcal{E}^{\parallel} (\{ P_{A,i}^{(\lambda)} \}_i))$ for different $\{ P_{A,i}^{(\lambda)} \}_i$ is contained in an interval that converges on $P_g(\overline{\mathcal{E}}_{SA'}(\tau))$ (See SM for proof). {If we define $P_{g}^{\parallel(\lambda)}\equiv\max_{\{ P_{A,i}^{(\lambda)} \}_i}P_g(\mathcal{E}^{\parallel}{ (\{ P_{A,i}^{(\lambda)} \}_i )})$ and $P_{g}^{\perp(\lambda)} \equiv\max_{\{ P_{A,i}^{(\lambda)} \}_i} P_g(\mathcal{E}^{\perp}{ (\{ P_{A,i}^{(\lambda)} \}_i )})$,}
it holds that
\begin{equation}
{\forall \delta >0,\, \exists \lambda_\delta >0 :  \, P_{g}^{\parallel(\lambda)} - P_g(\overline{\mathcal{E}}_{SA'}(\tau)) < \delta \,\, , \, \forall \lambda \in (\lambda_\delta,1) \, .}
\end{equation}
Hence, for $\overline{\delta} \equiv P_g( \overline{\mathcal{E}}_{SA'}(\tau+\Delta \tau)) - P_g( \overline{\mathcal{E}}_{SA'}(\tau))>0 $ (which is in the form of Eq. (\ref{PgBD})), there exists $\overline \lambda\in [0,1)$ such that
$P_{g}^{\parallel(\lambda)  } -P_g(\overline{\mathcal{E}}_{SA'}(\tau)) < P_g( \overline{\mathcal{E}}_{SA'}(\tau+\Delta \tau)) - P_g( \overline{\mathcal{E}}_{SA'}(\tau)) 
$ for any $ \lambda \in ( \overline{\lambda},1)$. It follows that
\begin{equation}\label{coolestineq}
{P_g( \overline{\mathcal{E}}_{SA'}(\tau+\Delta \tau)) - P_{g}^{\parallel (\lambda)}  >0 \, , \,\, \forall \lambda \in ( \overline{\lambda},1) \, .}
\end{equation}
To conclude, we consider inequalities  (\ref{deltaC}) and (\ref{coolestineq}) for $\lambda\in (\overline \lambda, 1)$ and we obtain a backflow
$$
\Delta {C}_{A}^\oPP \geq P_g \left(\lrho(\tau+\Delta \tau) , \{ \ketbra{i}{i}_A\}_i \right) - P_g \left(\lrho(\tau) , \{ P_{A,i}^{(\lambda)} \}_i \right) 
$$
\begin{equation}\label{finalbackflow}
{\geq\! \lambda \left(1- P_{g}^{\perp (\lambda)} \right) + (1-\lambda) \left( P_g(\overline{\mathcal{E}}_{SA'}(\tau+\Delta \tau)) - P_{g}^{\parallel (\lambda)} \right) > 0  ,}
\end{equation}
if and only if there is no CPTP intermediate map $V_S(\tau+\Delta \tau,\tau)$ for {$\{\Lambda_S(t,t_0)\}_t$}.
\end{proof}

We showed that for every non-Markovian evolution there exist initial probe states $\lrho(t_0)$  that provide at least one backflow of the correlation measure $C_A^\oPP$ if and only if the dynamics is non-Markovian. 
The robustness of this backflow is provided by the following properties that are valid for any {$\{\Lambda_S(t,t_0)\}_t$} and $[\tau,\tau+\Delta \tau]$: the guessing probability $P_g(\rho_{AB}, \{ P_{A,i}\}_i  )$
  is a Lipschitz continuous function of $\rho_{AB}\in \mathcal{S}(\mathcal{H}_A\otimes\mathcal{H}_B)$ and POVMs $\{ P_{A,i}\}_i$ (See SM), $\Pi_A^\oPP (\rho_{AB}^{(\lambda)}(t))$ does not depend on $\lambda$ and $t$, and there  exists a continuous interval of values of $\lambda$ for which $\rho_{AB}^{(\lambda)}(t_0)$ allows backflows of $C^{\oPP}_A (\rho_{AB}^{(\lambda)}(t))$ when there is no CPTP intermediate map $V_S(\tau+\Delta \tau,\tau)$. Therefore, if we add small enough perturbations to $\rho_{AB}^{(\lambda)}(t_0)$ and the optimal $\oPP$-POVMs obtained by the maximization in Eq. (\ref{measure2}), we still obtain backflows of $C_A^\oPP$ for any non-Markovian dynamics.
Hence,  there exists a set of initial states with the same dimension as $\mathcal S(\mathcal{H}_{A}\otimes \mathcal{H}_{B})$ that provide a backflow of $C^{\oPP}_A$ in the scenario described above (See SM for more details).

Since there are no particular assumptions for the structure of $\overline{\mathcal{E}}_{SA'}(t_0) $ \cite{BD}, it is straightforward to adapt our technique to any other ensemble. In particular, if the evolution of an initial ensemble $\{p_i,\phi_{SA',i}\}_{i=1}^n$ provides a backflow of $P_g(\{p_i,\phi_{SA',i}(t)\}_{i=1}^n)$ in a time interval $[\tau,\tau+\Delta \tau]$, we can consider $C^\PP_A(\psi_{AB}^{(\lambda)}(t_0))$, where $\PP=\{p_i\}_{i=1}^n$ and $\psi_{AB}^{(\lambda)}(t_0)= \sum_{i=1}^{n} {p}_i \, \ketbra{i}{i}_A \otimes \left( \lambda \sigma_{SA'} \otimes \ketbra{i}{i}_{A''} + (1-\lambda)\,  {\phi}_{SA',i} \otimes \ketbra{n+1}{n+1} \right) $ and obtain a backflow of $C^\PP_A(\psi_{AB}^{(\lambda)}(t))$ in $[\tau,\tau+\Delta \tau]$.
We make some examples of ensembles (different from $\overline{\mathcal{E}}_{SA'}(t_0)$) that can be considered to witness particular classes of non-Markovian evolutions. A constructive method that provides ensembles of two equiprobable states that witness any bijective or pointwise non-bijective non-Markovian dynamics is given in Ref. \cite{bogna}. The existence of two-state ensembles that detect any image non-increasing evolution, namely such that $\mbox{Im}(\Lambda_t) \subseteq \mbox{Im}(\Lambda_s)$  for any $s<t$, is proven in Ref. \cite{imagenon}. Finally, in Ref. \cite{qubitchrusc} is proven that two-state ensembles are sufficient to witness any non-Markovian qubit evolution.

Similarly to prior measures of non-Markovianity that catch increases of quantities that are monotonically decreasing under Markovian evolutions \cite{RHP,BLP,LFS,BognaChannel0,Volume}, we define the class
\begin{equation}\label{measure}
N^\PP(\{\Lambda_S(t,t_0)\}_t) \equiv \sup_{\rho_{ASA'}(t_0)} \int_{\frac{d}{dt}{C}_A^{\PP}(\rho_{ASA'}(t))>0} \frac{d}{dt}{C}_A^{\PP}(\rho_{ASA'}(t)) dt \, ,
\end{equation}
where the sup is over the possible ancillary systems ($A$ and $A'$) and the initial states $\rho_{ASA'}(t_0)\in \mathcal{S}(\mathcal{H}_A \otimes\mathcal{H}_S \otimes\mathcal{H}_{A'} )$. As a consequence of Theorem 1, if ${C}_A^{\PP}(\rho_{ASA'}(t))$ is differentiable, $N^\oPP(\{\Lambda_S(t,t_0)\}_t)>0$ if and only if the evolution is non-Markovian (See SM for details and a discussion of the non-differentiable case). 
Indeed, for any time interval where the evolution cannot be described by a CPTP intermediate map, we proved the existence of a set of initial states that show an increase of $C_A^\oPP$ in the same time interval. 
We notice that $N^\PP$ with $\PP=\{1/2,1/2\}$ is non-zero for any bijective or pointwise non-bijective non-Markovian evolution \cite{PRA}.

{ \textit{ Discussion.---}}In this work we showed that any non-Markovian dynamics can be witnessed through backflows of $C_A^\oPP$. For this purpose, we introduced a class of initial probe states  $\rho_{AB}^{(\lambda)}(t_0)$ that allows to accomplish this task. Hence, we proved the first one-to-one correspondence between CP-divisibility of evolutions, namely Markovianity, and the absence of correlation backflows.

It would be useful to obtain a constructive method that provides the elements of $\overline{\mathcal E}_{SA'}(t_0)$ that we used to define the initial probe state. Moreover, since the class of bipartite correlations that we studied does not consider the subsystems $A$ and $B$ symmetrically, an open question is to understand if also $C_{AB}^\oPP$ (see Eq. (\ref{measuresymm})) is able to witness any non-Markovian evolution.

The computation required to evaluate the measures of non-Markovianity $N^\PP$ can be significantly demanding. We consider interesting the possibility to formulate simplified versions of these measures (e.g, that require a simplified computation, are specialized to measure evolutions with particular properties).

\begin{acknowledgments}

{\textit{ Acknowledgments.---}} We would like to thank A. Ac{\'\i}n, B. Bylicka and M. Lostaglio for insightful discussions and comments on a previous draft.
Support from the 
Spanish MINECO (QIBEQI FIS2016-80773-P and Severo Ochoa SEV-2015-0522), the
Fundaci\'o Privada Cellex, the
Generalitat de Catalunya (CERCA Program and SGR1381), is acknowledged. D.D.S acknowledge support from the ICFOstepstone programme, funded by the Marie Sk\l odowska-Curie COFUND action (GA665884).

\end{acknowledgments}

 \begin{appendix}

\begin{widetext}

\vspace{0.5cm} 
\begin{center}
{\bf{Equivalence between non-Markovian dynamics and correlation backflows: Supplemental Material}}
\end{center}

\section{Monotonic behavior of $C^\PP_A$ under local operations}\label{LOCC}

We consider a general bipartite finite-dimensional quantum system with Hilbert space $\mathcal{H}_{AB}=\mathcal{H}_A\otimes \mathcal{H}_B$. Therefore, the states that we consider are $\rho_{AB}\in \mathcal S(\mathcal{H}_{AB})$.
We consider a generic finite probability distribution $\PP=\{p_i\}_{i=1}^n$ and we prove that {$C^\PP_A$} is monotone under local operations of the form $\Lambda_A \otimes  I_B$ and $I_A \otimes  \Lambda_B$ on $\rho_{AB}$, where $\Lambda_A$ ($\Lambda_B$) is a CPTP map on $A$ ($B$) and $I_A$ ($I_B$) is the identity map on $\mathcal{S}(\mathcal{H}_A)$ ($\mathcal{S}(\mathcal{H}_B)$).

In order to show the effect  of the application of a local operation of the form  $\Lambda_A \otimes  I_B$ on $C_A^\PP(\rho_{AB})$, we look at $\Pi^\PP_A(\rho_{AB})$ in a different way. Each element of this collection is a $\PP$-POVM for $\rho_{AB}$, i.e., they generate output ensembles where the output probability distribution is $\PP=\{p_i\}_i$. In fact, we can consider $C^\PP_A(\rho_{AB})$  as the maximization over all the possible output ensembles with output probability distribution  $\PP$ that we can generate measuring the subsystem $A$ of $\rho_{AB}$.

The effect of the first local operation that we consider is: $\tilde{\rho}_{AB}=  \Lambda_A \otimes  I_B \, ( \rho_{AB} ) = \sum_k \left( E_k \otimes \mathbbm{1}_B \right)  \rho_{AB}  \left( E_k \otimes \mathbbm{1}_B \right)^\dagger  \, ,$
where $\left\{ E_k \right\}_k$ is a set of Kraus operators that corresponds to $\Lambda_A$. Now we analyze the relation between $\Pi^\PP_A (\rho_{AB}) $ and $\Pi^\PP_A (\tilde{\rho}_{AB})$. Given a $\PP$-POVM for $\tilde \rho_{AB}$, i.e., $\{ P_{A,i}\}_i \in \Pi^\PP_A (\tilde{\rho}_{AB})$, the probabilities and the states of the output ensemble $\mathcal{E} \left( \tilde{\rho}_{AB}, \{ P_{A,i} \}_i\right)$ are
  $ \tr{ \tilde{\rho}_{AB}  P_{A,i} \otimes \mathbbm{1}_B }=p_i$ and $\tilde \rho_{B,i} = \trA{\tilde\rho_{AB}  P_{A,i}\otimes \mathbbm{1}_B }/{p}_i$. Now we write the $i$-th element of the output probability distribution that we obtain applying $\{ P_{A,i}\}_i$ on $\tilde \rho_{AB}$, namely $p_i = \tr{\Lambda_{A}\otimes I_B\, ( \rho_{AB})  \, P_{A,i}\otimes \mathbbm{1}_B }$, as follows
\begin{equation}
\tr{\sum_k (E_k \otimes \mathbbm{1}_B) \rho_{AB}  (E_k^\dagger \otimes \mathbbm{1}_B) P_{A,i}\otimes \mathbbm{1}_B } 
= \tr{\rho_{AB} \sum_k (E_k^\dagger \otimes \mathbbm{1}_B)  P_{A,i}  (E_k \otimes \mathbbm{1}_B) } 
= \tr{\rho_{AB}  \Lambda^*_{A} (P_{A,i} ) \otimes \mathbbm{1}_B } = \tr{\rho_{AB}  \tilde P_{A,i}\otimes  \mathbbm{1}_B }  ,
\end{equation}
where we have defined the operators $\tilde P_{A,i} \equiv \Lambda^*_{A} (P_{A,i} ) =  \sum_k (E_k^\dagger \otimes \mathbbm{1}_B)  P_{A,i}  (E_k \otimes \mathbbm{1}_B)$. Similary, we can write $\tilde{\rho}_{B,i}=\mbox{Tr}_A[ \tilde\rho_{AB}  P_{A,i} ] / p_i = \mbox{Tr}_A[ \rho_{AB}  \tilde P_{A,i} ] / p_i $. Therefore, since 
$p_i= \mbox{Tr}[ \rho_{AB}  \tilde P_{A,i} ] $ and $\tilde \rho_{B,i}=\mbox{Tr}_A[ \rho_{AB}  \tilde P_{A,i} ] / p_i$, if we apply $\{ \tilde P_{A,i} \}_i$ on $\rho_{AB}$ we obtain the same $\PP$-distributed output ensemble $\{p_i,\tilde \rho_{B,i}\}_i$ that we obtain applying $\{P_{A,i}\}_i$ on $\tilde \rho_{AB}$.
Next we show that:
$\{ \tilde{P}_{A,i} \}_i =\left\{ \Lambda^*_A \left( P_{A,i} \right) \right\}_i = \{ \sum_k E_k^\dagger P_{A,i}  E_k \}_i \, ,
$
is a proper $n$-output POVM. First, the elements of $\{ \tilde{P}_{A,i} \}_i$ sum up to the identity: 
$
\sum_i \tilde{P}_{A,i} =   \sum_{k,i} E_k^\dagger \, P_{A,i} \,  E_k = \sum_{k} E_k^\dagger \, \left( \sum_i P_{A,i} \right) \, E_k = \sum_{k} E_k^\dagger\,  E_k = \mathbbm{1}_B \, .
$
Moreover, we show that they are positive semi-definite operators. Indeed, for any $\ket{\psi}_A\in \mathcal{H}_A$, we have $\bra{\psi}_A\tilde{P}_{A,i}  \ket{\psi}_A=\sum_k (\bra{\psi}_A E_k^\dagger) \,  P_{A,i} \, ( E_k  \ket{\psi}_A) = \sum_k \bra{\psi_k}_A  P_{A,i} \ket{\psi_k}_A\geq 0$, where each element of the last sum is non-negative because $P_{A,i}$ is positive semi-definite.
 It follows that $\{ \tilde{P}_{A,i} \}_i$ is a POVM and in particular a $\PP$-POVM for $\rho_{AB}$, i.e., $\{ \tilde{P}_{A,i} \}_i\in \Pi^\PP_A(\rho_{AB})$.  Thus, for every $\PP$-POVM $\{P_{A,i}\}_i \in \Pi^\PP_A(\tilde\rho_{AB})$ for $\tilde \rho_{AB}$, there is a $\PP$-POVM $\{\tilde P_{A,i}\}_i \in \Pi^\PP_A(\rho_{AB})$ for $\rho_{AB}$, such that the output ensembles are identical: $\mathcal{E} (\tilde \rho_{AB}, \{P_{A,i}\}_i) = \mathcal{E} (\rho_{AB}, \{\tilde P_{A,i} \}_i )$. Hence, any $\PP$-distributed ensemble of $B$ that can be generated from $\tilde{\rho}_{AB}$ can also be obtained from $\rho_{AB}$. Therefore, we obtain the following inclusion 
\begin{equation}\label{EI}
\bigcup_{ \{P_{A,i} \}_i \in \Pi^\PP_A (\tilde{\rho}_{AB} ) }   \!\!\!\!\!\!  \mathcal{E} \left( \tilde{\rho}_{AB}, \, \{ P_{A,i} \}_i \right)\,  \subseteq \bigcup_{ \{P_{A,i} \}_i \in \Pi^\PP_A ({\rho}_{AB} ) }    \!\!\!\!\!\! \mathcal{E} \left( {\rho}_{AB}, \, \{ P_{A,i} \}_i \right) \, .
\end{equation}
Finally, since as we said above $C^\PP_A(\rho_{AB})$ is the maximum guessing probability of the $\PP$-distributed output ensembles that can be generated from $\rho_{AB}$, from Eq. (\ref{EI}) we conclude that  $C^\PP_A(\rho_{AB})$ is defined as a maximization over a set that includes the set over which maximization defines $C^\PP_A(\tilde\rho_{AB})$. Hence, for any state $\rho_{AB}$ and CPTP map $\Lambda_A$, we obtain
\begin{equation}\label{monA2}
C^\PP_A \left( \rho_{AB} \right) \geq C^\PP_A \left(  \Lambda_A \otimes  I_B \,  ( \rho_{AB}) \right)  .
\end{equation}

Next we show that $C^\PP_A(\rho_{AB})$ is monotonic under local operations of the form $I_A \otimes  \Lambda_B$. We find that the collection of the $\PP$-POVMs for $\tilde{\rho}_{AB} =  I_A \otimes  \Lambda_B \, ( {\rho}_{AB})$, namely $\Pi^\PP_A (\tilde{\rho}_{AB})$, coincides with $\Pi^\PP_A (\rho_{AB})$. 
In order to prove this, we apply a general POVM $\{P_{A,i}\}_i$ both on $\rho_{AB}$ and $\tilde \rho_{AB}$ and we show that the respective output ensembles are defined by the same probability distribution. Indeed, being $\tr{\rho_{AB} P_{A,i}}$ ($\tr{ I_A\otimes \Lambda_B \, ( \rho_{AB})  P_{A,i} }$) the probability for the $i$-th output of the POVM considered when it is applied on $\rho_{AB}$ ($\tilde\rho_{AB}$), we have $\tr{ I_A\otimes \Lambda_B \, ( \rho_{AB})  P_{A,i} } = \tr{\rho_{AB}  P_{A,i} }$, where this identity uses the trace-preserving property of the superoperator $I_A\otimes \Lambda_B$. Consequently, if $\{P_{A,i}\}_i$ is a $\PP$-POVM for $\rho_{AB}$, which means that $\tr{\rho_{AB} P_{A,i}}=p_i$, in the same way $\tr{ I_A\otimes \Lambda_B \, ( \rho_{AB})  P_{A,i} }=p_i$. Hence, $\{ P_{A,i}\}_i\in \Pi^\PP_A(\rho_{AB})$ if and only if $\{ P_{A,i}\}_i\in \Pi^\PP_A(\tilde \rho_{AB})$, i.e.,
\begin{equation}\label{uguale}
\Pi^\PP_A (\rho_{AB} ) = \Pi^\PP_A (\tilde{\rho}_{AB}) \, .
\end{equation}
Given a $\PP$-POVM $\{P_{A,i}\}_i$ both for $\rho_{AB}$ and $\tilde{\rho}_{AB}$, we compare the corresponding output states
\begin{equation}\label{rhobicontratti}
\tilde{\rho}_{B,i} = \Lambda_B \, ( \trA{ \rho_{AB} P_{A,i} \otimes \mathbbm{1}_B  }  / p_i ) = \Lambda_B (\rho_{B,i}) \, .
\end{equation}
From Eq. (\ref{rhobicontratti}) and the definition of the guessing probability, it follows that
\begin{equation}\label{monB0}
P_g\left(  \left\{ p_i, \, \rho_{B,i}  \right\}_i  \right) \geq P_g\left(  \left\{ p_i, \, \Lambda_B (\rho_{B,i})  \right\}_i  \right) \, .
\end{equation}
The consequence of the last relation is that for any $\PP$-distributed output ensemble ensemble that we can generate from $\tilde \rho_{AB}$ there exists at least one $\PP$-distributed output ensemble that we can generate from $\rho_{AB}$ for which the guessing probability is equal or greater. Hence, considering the definition of $C^\PP_A$, Eqs. (\ref{uguale}) and (\ref{monB0}), we conclude that
\begin{equation}\label{monB}
C_A^\PP \left( \rho_{AB} \right) \geq C_A^\PP \left(  I_A \otimes \Lambda_B   \,  ( \rho_{AB}) \right) \, ,
\end{equation}
for any state $\rho_{AB}$ and CPTP map $\Lambda_B$.

\section{Performing $\oPP$-POVMs on the probe state: the orthogonal and the parallel components}

In this section we prove that, if we apply the projective $\oPP$-POVM $\{ \ketbra{i}{i}_A\}_i $ on $A$ for $\lrho(t)$, we obtain
\begin{equation}\label{PgApix}
P_g \left(\lrho(t) , \{ \ketbra{i}{i}_A\}_i  \right) = \lambda + (1-\lambda) \, P_g\left(\overline{\mathcal{E}}_{SA'}(t) \right)  .
\end{equation}
Moreover, for a general $\oPP$-POVM on $A$ for $\lrho(t)$ different from $\{ \ketbra{i}{i}_A\}_i $, we have
\begin{equation}\label{generalPg}
P_g \left(\lrho(t) , \{P_{A,i} \}_i \right) \! = \! \lambda P_g \left( \{ \overline{p}_i , \rho_{A'',i}^{\perp }  \}_i \right) + (1-\lambda)  P_g \left( \{ \overline{p}_i , \rho_{SA',i}^{\parallel}(t)  \}_i \right) ,
\end{equation}
for some $\{ \rho_{A'',i}^{\perp }  \}_i$ and $\{ \rho_{SA',i}^{\parallel}(t)  \}_i$ that we define.
First, we notice that the projective measurement $\{ \ketbra{i}{i}_A\}_{i=1}^{\overline{n}} $  is a $\oPP$-POVM on $A$ for $\lrho (t)$  for any $t$ and  $\lambda$. We consider $\mathcal{E} (\lrho(t) , \{ \ketbra{i}{i}_A\}_i )$, namely the ensemble of $B$ that we obtain measuring $\lrho(t)$ with $\{ \ketbra{i}{i}_A\}_i$:
\begin{equation}\label{EnsPi}
\mathcal{E} (\lrho(t) , \{ \ketbra{i}{i}_A\}_i )= \left\{\overline{p}_i, \lambda \, \sigma_{SA'} (t) \otimes \ketbra{i}{i}_{A''} + (1-\lambda) \overline{\rho}_{B,i} (t) \right\}_{i=1}^{\overline{n}}  \, ,
\end{equation}
where $\overline{\rho}_{B,i}  =    \overline{\rho}_{SA',i}\otimes \ketbra{\overline{n}+1}{\overline{n}+1}_{A''}$.
We evaluate the guessing probability of this ensemble and we obtain
$$
P_g (\lrho(t) , \{ \ketbra{i}{i}_A\}_i )  =  \max_{ \{ P_{B,i} \}_i } \sum_{i=1}^{\overline{n}} \overline{p}_i \, \trB{\left( \lambda\sigma_{SA'}(t) \otimes \ketbra{i}{i}_{A''} +  (1-\lambda) \overline{\rho}_{SA',i} (t)  \otimes \ketbra{\overline{n}+1}{\overline{n}+1}_{A''} \right)  P_{B,i} } 
$$
\begin{equation}\label{perppar0}
=\max_{ \{ P_{B,i} \}_i } \left( \lambda \sum_{i=1}^{\overline{n}} \overline{p}_i \, \trB{ \sigma_{SA'}(t) \otimes \ketbra{i}{i}_{A''}    \, P_{B,i} } + (1-\lambda) \sum_{i=1}^{\overline{n}}  \overline{p}_i \trB{\overline{\rho}_{SA',i} (t)  \otimes \ketbra{\overline{n}+1}{\overline{n}+1}_{A''} \,  P_{B,i} } \right)  .
\end{equation}

We notice that, for any $i=1,\dots, \overline n$, every state that belongs to the set $\{\sigma_{SA'}(t) \otimes \ketbra{i}{i}_{A''}    \}_i$  is orthogonal to every state of the set $\{\overline{\rho}_{SA',i} (t)  \otimes \ketbra{\overline{n}+1}{\overline{n}+1}_{A''}\}_i$. It follows that, for any $i=1,\dots, \overline n$, the value of $\trB{\sigma_{SA'}(t) \otimes \ketbra{i}{i}_{A''}    \, P_{B,i}}$ depends only on the components of $P_{B,i}$ that belong to span($\{ \ketbra{i}{j}_B \}_{ij}$), where $\ket{i}_B$ and $\ket{j}_B $ belong to the tensor product between the elements of $\mathcal{M}_{SA'}$, i.e., an orthonormal basis of $\mathcal{H}_{S}\otimes \mathcal{H}_{A'}$, and $\{ \ket{k}_{A''} \}_{k=1}^{\overline n}$ (notice that dim($\mathcal{H}_{A''}) =\overline n +1$). Similarly, for any $i=1,\dots, \overline n$, the value of $\trB{\overline{\rho}_{SA',i} (t)  \otimes \ketbra{\overline{n}+1}{\overline{n}+1}_{A''} \,  P_{B,i}}$ depends only on the components of $P_{B,i}$ that belong to span($\{ \ketbra{i'}{j'}_B \}_{i'j'}$), where $\ket{i'}_B$ and $\ket{j'}_B $ belong to the tensor product between the elements of $\mathcal{M}_{SA'}$ and $ \ket{\overline{n}+1}_{A''} $. We further note that no operator defined on span($\{ \ketbra{i}{j}_B \}_{ij}$)$\oplus$span($\{ \ketbra{i'}{j'}_B \}_{i'j'}$) that is not positive semidefinite can be made positive semidefinite by adding something outside span($\{ \ketbra{i}{j}_B \}_{ij}$)$\oplus$span($\{ \ketbra{i'}{j'}_B \}_{i'j'}$). Therefore, we can limit the maximization in Eq. (\ref{perppar0}) to be over POVMs $P_{B,i}$ that are defined on span($\{ \ketbra{i}{j}_B \}_{ij}$)$\oplus$span($\{ \ketbra{i'}{j'}_B \}_{i'j'}$), without affecting the optimal value. 
Since span($\{ \ketbra{i}{j}_B \}_{ij}$) is orthogonal to span($\{ \ketbra{i'}{j'}_B \}_{i'\!j'}$), the maximization in Eq. (\ref{perppar0}) can be divided in two independent maximizations
$$
P_g (\lrho(t) ,\{ \ketbra{i}{i}_A\}_i   ) =\lambda \max_{ \{ P_{B,i} \}_i }  \sum_{i=1}^{\overline{n}} \overline{p}_i \, \trB{ \sigma_{SA'}(t) \otimes \ketbra{i}{i}_{A''}    \, P_{B,i} } + (1-\lambda) \max_{ \{ P_{B,i} \}_i } \sum_{i=1}^{\overline{n}}  \overline{p}_i \trB{\overline{\rho}_{SA',i} (t)  \otimes \ketbra{\overline{n}+1}{\overline{n}+1}_{A''} \,  P_{B,i} } 
$$
$$
=\lambda P_g( \{\overline p_i , \sigma_{SA'}(t) \otimes \ketbra{i}{i}_{A''} \}_i) + (1-\lambda)  P_g( \{\overline p_i , \overline{\rho}_{SA',i} (t)  \otimes \ketbra{\overline{n}+1}{\overline{n}+1}_{A''} \}_i )
$$
\begin{equation}\label{PgApi}
=\lambda P_g( \{\overline p_i ,  \ketbra{i}{i}_{A''} \}_i) + (1-\lambda)  P_g( \{\overline p_i , \overline{\rho}_{SA',i} (t)   \}_i ) = \lambda + (1-\lambda) P_g(\overline{\mathcal{E}}_{SA'}(t) ) \, , 
\end{equation} 
where we have used $P_g ( \{ \overline{p}_i , \ketbra{i}{i}_{A''} \}_i ) =1$, namely the possibility to perfectly  distinguish ensembles of orthonormal states, $ P_g( \{\overline p_i , \sigma_{SA'}(t) \otimes \ketbra{i}{i}_{A''} \}_i)= P_g( \{\overline p_i ,  \ketbra{i}{i}_{A''} \}_i)$ and $ P_g( \{\overline p_i , \overline{\rho}_{SA',i} (t)  \otimes \ketbra{\overline{n}+1}{\overline{n}+1}_{A''} \}_i ) =  P_g( \{\overline p_i , \overline{\rho}_{SA',i} (t)  \}_i )=P_g(\overline{\mathcal{E}}_{SA'}(t) )$.

The output ensemble that we obtain applying a generic $\oPP$-POVM $ \{ P_{A,i} \}_i $ on $A$ for $\lrho(t)$ different from $  \{ \ketbra{i}{i}_A\}_i $  is $\mathcal{E} (\lrho (t), \{ P_{A,i} \}_i )$. The $k$-th state of this ensemble is
$$
  \rho^{(\lambda)}_{B,k}(t)= \frac{ \trA{\lrho(t)   P_{A,k}   \otimes \mathbbm{1}_B } }{\overline{p}_k} =
 \sum_{i=1}^{\overline{n}}  \frac{ \overline{p}_i}{\overline{p}_k}   \, \trA{ \ketbra{i}{i}_A  P_{A,k} } \left(\lambda\sigma_{SA'}(t) \otimes \ketbra{i}{i}_{A''} +  (1-\lambda) \overline{\rho}_{B,i} (t)   \right) $$
\begin{equation}
= \sum_{i=1}^{\overline{n}}  \frac{ \overline{p}_i \left(  P_{A,k} \right)_{ ii} }{\overline{p}_k}  \left( \lambda\sigma_{SA'}(t) \otimes \ketbra{i}{i}_{A''}+ (1-\lambda) \overline{\rho}_{B,i} (t) \right)  \, ,
\end{equation}
where $( P_{A,k} )_{ii} = \bra{i}_A P_{A,k} \ket{i}_A  \geq 0$ is the $i$-th diagonal element of $P_{A,k}$ in the basis $\mathcal{M}_A = \{\ket{i}_A\}_{i=1}^{\overline n}$. Keeping in mind that $\oPP$ is a finite probability distribution and $\overline{p}_k > 0$ for any $k$, we define the parameters $e_{i k}\equiv ( P_{A,k} )_{ii}\overline{p}_i/\overline{p}_k \geq 0$. Since  $\rho_{B,k}^{(\lambda)} (t)$ and the states $\lambda\sigma_{SA'}(t) \otimes \ketbra{i}{i}_{A''}+ (1-\lambda) \overline{\rho}_{B,i} (t) $ are trace one operators for any $i=1,\dots,\overline n$ , we conclude that  $\sum_i e_{i k} =1$ for any $k=1,\dots,\overline n$. Therefore,  $\{e_{ik}\}_{i=1}^{\overline n}$ is an $\overline n$-element probability distribution for any value of $k=1,  ..., \overline n$. We write:
\begin{equation}\label{ll}
\rho_{B,k}^{(\lambda)} (t) = \sum_{i=1}^{\overline{n}}e_{ik} \left( \lambda\sigma_{SA'}(t) \otimes \ketbra{i}{i}_{A''} + (1-\lambda) \overline{\rho}_{B,i} (t) \right) 
 = \lambda \sum_{i=1}^{\overline{n}} e_{ik} \sigma_{SA'}(t) \otimes \ketbra{i}{i}_{A''} + (1-\lambda)\sum_{i=1}^{\overline{n}}e_{ik}  \overline{\rho}_{B,i} (t) = \lambda \sigma_{B,k}^{\perp }(t) + (1-\lambda) \sigma_{B,k}^{\parallel } (t) ,
\end{equation}
where we have used the definitions
\begin{equation}\label{sigmaperp}
\sigma_{B,k}^{\perp } (t) \equiv \sigma_{SA'}(t) \otimes \left( \sum_{i=1}^{\overline{n}} e_{ik}  \ketbra{i}{i}_{A''} \right) \equiv \sigma_{SA'}(t) \otimes \rho_{A'',k}^{\perp } \, ,
\end{equation}
\begin{equation}\label{sigmaparallel}
\sigma_{B,k}^{\parallel } (t) \equiv  \left( \sum_{i=1}^{\overline{n}}e_{ik} \, \overline{\rho}_{SA',i}  (t) \right) \otimes \ketbra{\overline{n}+1}{\overline{n}+1}_{A''} \equiv \rho_{SA',k}^{\parallel } (t)  \otimes \ketbra{\overline{n}+1}{\overline{n}+1}_{A''} .
\end{equation}
Each state $\rho_{A'',k}^{\perp }$  ($\rho_{SA',k}^{\parallel }(t)$)  is a convex combination of the states $\{  \ketbra{i}{i}_{A''} \}_{i=1}^{\overline n}$ ($\{\overline{\rho}_{SA',i}(t) \}_{i=1}^{\overline{n}}$) that does not depend on $\lambda$ but depends on the $\oPP$-POVM $\{P_{A,i}\}_i$ chosen. From Eq. (\ref{ll}) it follows that, if we consider a generic $\oPP$-POVM $\{P_{A,i} \}_i$ for $\lrho(t)$, we obtain
\begin{equation}\label{EnsJ}
\mathcal{E} (\lrho (t), \{P_{A,i} \}_i)= \{ \overline{p}_i, \lambda \sigma_{B,i}^{\perp }(t) + (1-\lambda) \sigma_{B,i}^{\parallel } (t)\}_i \, ,
\end{equation}
and therefore, similarly to Eq. (\ref{PgApi}), now we can write
\begin{equation}\label{boh4}
P_g(\lrho(t) , \{P_{A,i} \}_i ) = 
\lambda P_g ( \{ \overline{p}_i , \rho_{A'',i}^{\perp } \}_i ) + (1-\lambda)  P_g ( \{ \overline{p}_i , \rho_{SA',i}^{\parallel}(t)  \}_i ) \, .
\end{equation}

\section{Analysis of case (A)}\label{caseA}

Let assume that for some $\alpha\in[0,1)$ we have that $\{ P_{A,i}^{(\alpha)}\}_i = \{ \ketbra{i}{i}_A \}_i$, i.e., this projective measurement is one of the optimal $\oPP$-POVM that accomplishes the maximization for $C^\oPP_A(\rho^{(\alpha)}_{AB}(\tau))$, and that for some $\beta> \alpha$ instead we have that $\{ \ketbra{i}{i}_A \}_i$ is not optimal. In this section we show that these two assumptions are incompatible and lead to a contradiction. The first condition implies that, when $\lambda=\alpha$ the optimal $\oPP$-POVM that provides the greatest value of $P_g(\rho_{AB}^{(\alpha)}(\tau), \{P_{A,i}\}_i)$  is  $\{ P_{A,i}^{(\alpha)}\}_i=\{ \ketbra{i}{i}_A \}_i$ and therefore 
$$
\alpha + (1-\alpha) P_g(\overline{\mathcal{E}}_{SA'}(\tau)) \geq  \alpha P_g (\mathcal{E}^{\perp (\beta)}) + (1-\alpha) P_g  (\mathcal{E}^{\parallel (\beta)})  \, ,
$$
\begin{equation}\label{alpha}
\alpha \left(1- P_g (\mathcal{E}^{\perp (\beta)}) \right) + (1-\alpha) \left(P_g(\overline{\mathcal{E}}_{SA'}(\tau)) -  P_g  (\mathcal{E}^{\parallel (\beta)}) \right) \geq 0 \, , 
\end{equation}
where we also considered the cases where  $\{ P_{A,i}^{(\beta)}\}_i $ is optimal both for $\lambda=\alpha$ and $\lambda=\beta$.
 On the other hand, for $\lambda=\beta>\alpha$ we have that $\{ \ketbra{i}{i}_A \}_i$ is not an optimal $\oPP$-POVM for the maximization needed for $C^\oPP_A(\rho^{(\beta)}_{AB}(\tau))$ and
\begin{equation}\label{beta}
\beta P_g (\mathcal{E}^{\perp (\beta)}) + (1-\beta) P_g  (\mathcal{E}^{\parallel (\beta)})  > \beta + (1-\beta) P_g(\overline{\mathcal{E}}_{SA'}(\tau))  \, ,
\end{equation}
which can be written as
\begin{equation}\label{48}
\beta \left(P_g (\mathcal{E}^{\perp (\beta)}) -1 +P_g(\overline{\mathcal{E}}_{SA'}(\tau))- P_g (\mathcal{E}^{\parallel (\beta)}) \right)   > P_g(\overline{\mathcal{E}}_{SA'}(\tau)) -P_g(\mathcal{E}^{\parallel (\beta)} ) \, ,
\end{equation}
and therefore, subtracting the quantity $\alpha \left(P_g (\mathcal{E}^{\perp (\beta)}) -1 +P_g(\overline{\mathcal{E}}_{SA'}(\tau))- P_g (\mathcal{E}^{\parallel (\beta)}) \right)$ from each side of inequality (\ref{48}), we obtain
\begin{equation}\label{beta2}
(\beta - \alpha) \left(P_g (\mathcal{E}^{\perp (\beta)}) -1 +P_g(\overline{\mathcal{E}}_{SA'}(\tau))- P_g (\mathcal{E}^{\parallel (\beta)}) \right)  > \alpha \left(1- P_g (\mathcal{E}^{\perp (\beta)}) \right) + (1-\alpha) \left( P_g(\overline{\mathcal{E}}_{SA'}(\tau))  - P_g  (\mathcal{E}^{\parallel (\beta)}) \right) \, .
\end{equation}
If inequality (\ref{beta}) holds, then $P_g  (\mathcal{E}^{\parallel (\beta)}) >  P_g(\overline{\mathcal{E}}_{SA'}(\tau)) $. Therefore, $  P_g(\overline{\mathcal{E}}_{SA'}(\tau)) -P_g  (\mathcal{E}^{\parallel (\beta)})<0$ and we conclude that the left-hand side of inequality (\ref{beta2}) is negative. The right-hand side of the same inequality is instead non-negative for inequality (\ref{alpha}). This contradiction shows that if for some value of the parameter $\lambda$ the orthogonal measurement $\{ \ketbra{i}{i}_A \}_i$ maximizes $P_g(\lrho(\tau), \{P_{A,i}\}_i)$, then it is also the case for any greater value of $\lambda$. In conclusion, if one of the optimal measurement is $\{ \ketbra{i}{i}_A \}_i$ for $\lambda=\alpha$, the same is true for any $\beta \in [\alpha,1)$.

\section{Study of the limit $\lambda \rightarrow 1$ in case (B)}\label{appD}
First, we notice that the set of $\oPP$-POVMs on $A$ for $\lrho (t)$ is a set that does not depend on $\lambda$ and $t$. Indeed, we use the notation $ \Pi_A^{\oPP}= \Pi_A^\oPP (\lrho (\tau))$. 

Now we prove that the only optimal $\oPP$-POVM for $C^\oPP_A(\rho_{AB}^{(1)}(\tau))$ is the projective measurement $\{ \kbra{i}_A \}_i$. In the case of an optimal $\{P_{A,i}\}_i\in \Pi_A^\oPP$ for $\rho_{AB}^{(1)}(\tau)$ we obtain the output ensemble (see Eq. (\ref{sigmaperp}))
\begin{equation}
\mathcal{E}(\rho_{AB}^{(1)}(\tau), \PAi) = \{ \overline p_i ,\sigma_{SA'}(\tau) \otimes\sum_{j} e_{ji} \kbra{j}_{A''}    \}_i \, ,
\end{equation}
where $\sum_j e_{ji}=1$ for any $i=1,\dots,\overline n$. Since $P_g(\rho_{AB}^{(1)}(\tau), \{\kbra{i}_{A}\}_i)=1$, an optimal $\oPP$-POVM different from $\{ \kbra{i}_{A}\}_i$ must provide an output  ensemble $\mathcal{E}(\rho_{AB}^{(1)}(\tau), \PAi) $ of orthogonal states. Given the identity $P_g(\mathcal{E}(\rho_{AB}^{(1)}(\tau), \PAi) =P_g(  \{ \overline p_i ,\sum_{j} e_{ji} \kbra{j}_{A''}    \}_i)$, we have to check if, for some $e_{ij}$, the ensemble $\{\overline p_i , \sum_{j} e_{ji} \kbra{j}_{A''} \}_i   $  can be an orthogonal ensemble of states different from $\{\overline p_i\, \kbra{i}_{A''}\}_i$. Each state $\rho^\perp_{A'',i}=\sum_{j} e_{ji} \kbra{j}_{A''}$ is defined as a convex combination of the states $\{ \kbra{i}_{A''}\}_i$. Two such states are orthogonal only if the respective convex combinations do not have any element $\kbra{i}_{A''}$ in common. Therefore, the only way to have $\bar{n}$ orthogonal output states is if for each $i$ the state is of the form $\rho^\perp_{A'',i}= \kbra{j}_{A''}$ for some $j=j(i)$ exclusively assigned to $i$. Thus, each $P_{A,i}$ has only one nonzero diagonal element $( P_{A,i} )_{jj}= \bra{j}_A P_{A,i} \ket{j}_A$. Since $\sum_i P_{A,i}=\mathbbm{1}_A$ this is only possible if $\PAi=\{ \kbra{i}_A \}_i$.

We proved that $\{\kbra{i}_A\}_i \in \Pi_A^\oPP $ is the only optimal $\oPP$-POVM for the evaluation of $C^\oPP_A(\rho_{AB}^{(1)}(\tau))$. Therefore, for any $\oPP$-POVM  $\{ P_{A,i}\}_i \neq \{\kbra{i}_A \}_i$ we have that $P_g(\rho_{AB}^{(1)}(\tau), \{ P_{A,i}\}_i  ) < 1$.
We notice that the set $\Pi_A^\oPP$ is closed and bounded, i.e., it is compact. Indeed, it is a subset of $\mathcal{B}(\mathcal{H}_A)$ that is defined through linear constraints involving identities and relations of semi-positivity. The guessing probability $P_g(\rho_{AB}^{(1)}(\tau), \{ P_{A,i}\}_i  )$ is a continuous function on this compact set of $\oPP$-POVMs.

We now show that $P_g(\rho_{AB}^{(\lambda)}(\tau), \{ P_{A,i}\}_i  )$ is Lipschitz continuous in $\lambda$. In other words we construct a bound on the change of the guessing probability for a given change in $\lambda$. To do so we first show that $P_g(\rho_{AB}, \{ P_{A,i}\}_i  )$ is Lipschitz continuous on the set of states.
Consider $P_g(\rho_{AB}, \{ P_{A,i}\}_i  )$ as a function of $\rho_{AB}$.
We consider a pair $\rho_{AB}^{1}$, $\rho_{AB}^{2}$ and observe that

\begin{eqnarray}\label{firstineq}
\max_{\{ P_{B,i} \}_i } \sum_i\mbox{Tr}[ P_{A,i}\otimes P_{B,i} \rho_{AB}^{1}]&=&\max_{ \{ P_{B,i} \}_i } \sum_i\mbox{Tr}[ P_{A,i}\otimes P_{B,i} (\rho_{AB}^{2}+(\rho_{AB}^{1}-\rho_{AB}^{2})]\nonumber\\
&\leq &
\max_{\{ P_{B,i} \}_i } \sum_i\mbox{Tr}[ P_{A,i}\otimes P_{B,i} \rho_{AB}^{2}]+\max_{ \{ P_{B,i} \}_i } \sum_i\mbox{Tr}[ P_{A,i}\otimes P_{B,i} (\rho_{AB}^{1}-\rho_{AB}^{2})].
\end{eqnarray}
Let $\Delta$ be a diagonal matrix such that $\Delta= U(\rho_{AB}^{1}-\rho_{AB}^{2})U^\dagger$ for a unitary $U$. Let $\Delta_+$ and $\Delta_-$ be the two diagonal positive semidefinite matrices such that $\Delta=\Delta_+-\Delta_-$. Note that $U^\dagger\Delta_+U$ and $U^\dagger\Delta_-U$ are positive semidefinite. This implies

\begin{eqnarray}
\max_{ \{ P_{B,i} \}_i }\sum_i \mbox{Tr}[P_{A,i}\otimes P_{B,i} (\rho_{AB}^{1}-\rho_{AB}^{2})]&=&
\max_{ \{ P_{B,i} \}_i }\sum_i \mbox{Tr}[P_{A,i}\otimes P_{B,i} U^\dagger(\Delta_+-\Delta_-)U]\nonumber\\
&\leq &
\max_{ \{ P_{B,i} \}_i }\sum_i \mbox{Tr}[P_{A,i}\otimes P_{B,i}(U^\dagger\Delta_+U)]+\max_{ \{ P_{B,i} \}_i }\sum_i \mbox{Tr}[P_{A,i}\otimes P_{B,i} (U^\dagger\Delta_-U)].
\end{eqnarray}
Since POVM elements are positive semidefinite $\mbox{Tr}[P_{A,i}\otimes  P_{B,j} (U^\dagger\Delta_+U)]$ is positive for each pair $P_{A,i}$, $P_{B,j}$.
Therefore $\mbox{Tr}[\sum_iP_{A,i}\otimes P_{B,i} (U^\dagger\Delta_+U)]\leq\mbox{Tr}[\sum_iP_{A,i}\otimes \sum_j P_{B,j} (U^\dagger\Delta_+U)]=\mbox{Tr}[U^\dagger\Delta_+U]=\mbox{Tr}[\Delta_+]$. Likewise $\sum_i\mbox{Tr}[P_{A,i}\otimes P_{B,i} (U^\dagger\Delta_-U)]\leq \mbox{Tr}[\Delta_-]$. Thus,
\begin{equation}\label{thirdineq}
\max_{ \{ P_{B,i} \}_i }\sum_i \mbox{Tr}[P_{A,i}\otimes  P_{B,i}(\rho_{AB}^{1}-\rho_{AB}^{2})]\leq \mbox{Tr}[\Delta_++\Delta_-]=||\rho_{AB}^{1}-\rho_{AB}^{2}||_1.
\end{equation}
Considering Eqs. (\ref{firstineq}) and (\ref{thirdineq}) we can now conclude that
\begin{equation}
P_g(\rho_{AB}^{1}, \{ P_{A,i}\}_i  )-P_g(\rho_{AB}^{2}, \{ P_{A,i}\}_i  )\leq ||\rho_{AB}^{1}-\rho_{AB}^{2}||_1.
\end{equation}
By exchanging the $1$ and $2$ in the above derivation we obtain 
\begin{equation}
P_g(\rho_{AB}^{2}, \{ P_{A,i}\}_i  )-P_g(\rho_{AB}^{1}, \{ P_{A,i}\}_i  )\leq ||\rho_{AB}^{1}-\rho_{AB}^{2}||_1.
\end{equation}
Thus
\begin{equation}\label{lip1}
|P_g(\rho_{AB}^{1}, \{ P_{A,i}\}_i  )-P_g(\rho_{AB}^{2}, \{ P_{A,i}\}_i  )|\leq ||\rho_{AB}^{1}-\rho_{AB}^{2}||_1.
\end{equation}
Note that this bound is independent of $\{ P_{A,i}\}_i$. Thus we see that $P_g(\rho_{AB}, \{ P_{A,i}\}_i  )$ is Lipschitz continuous on the set of states.
Next we consider the pair $\rho_{AB}^{(\lambda_1)}(\tau),\rho_{AB}^{(\lambda_2)}(\tau)$ and note that the trace norm $||\rho_{AB}^{(\lambda_1)}(\tau)-\rho_{AB}^{(\lambda_2)}(\tau)||_1=2|\lambda_1-\lambda_2|$.
Therefore,
\begin{equation}\label{epp}
|P_g(\rho_{AB}^{(\lambda_1)}(\tau), \{ P_{A,i}\}_i  )-P_g(\rho_{AB}^{(\lambda_2)}(\tau), \{ P_{A,i}\}_i  )|\leq 2|\lambda_1-\lambda_2|.
\end{equation}
Thus we see that $P_g(\rho_{AB}^{(\lambda)}(\tau), \{ P_{A,i}\}_i  )$ is Lipschitz continuous in $\lambda$.

We next consider how the set of optimal $\oPP$-POVMs converges to $\{ \kbra{i}_{A}\}_i$ as $\lambda\to 1$ using the bound in Eq. (\ref{epp}).
Consider a semi-open neighbourhood $O_1$ of the projective $\oPP$-POVM $\{ \kbra{i}_{A}\}_i$ such that the set $S_1\equiv \Pi_A^{\oPP}-O_1$ of $\oPP$-POVMs not in $O_1$ is closed. Since the set $S_1$ is closed and bounded and $P_g(\rho_{AB}^{(1)}(\tau),\{ P_{A,i}\}_i)$ is a continuous function on $\Pi_A^{\oPP}$ there exists a maximum value $m_1<1$ of $P_g(\rho_{AB}^{(1)}(\tau), \{ P_{A,i}\}_i)$ on $S_1$, i.e.,
$m_1 \equiv \max_{\{P_{A,i}\}_i \in S_1} P_g(\rho_{AB}^{(1)}(\tau), \{P_{A,i}\}_i ) < 1 
$. 
Then, due to Eq. (\ref{epp}), for $\epsilon>0$ and $\lambda=1-\epsilon$ it holds that $P_g(\rho_{AB}^{(1-\epsilon)}(\tau), \{ P_{A,i}\}_i)\leq m_1+2\epsilon$ on $S_1$ and the maximum value of $P_g(\rho_{AB}^{(1-\epsilon)}(\tau), \{ P_{A,i}\}_i)$ on $O_1$ is larger or equal to $1- 2\epsilon$.
There exists a sufficiently small $\epsilon_1>0$ such that $1-2\epsilon_1= m_1+2\epsilon_1$.
For all $\epsilon<\epsilon_1$ the set of optimal $\oPP$-POVMs belongs to $O_1$.

We next consider a sequence of semi-open sets $O_i$ which all contain $\{ \kbra{i}_{A}\}_i$ and are such that $O_{i+1}\subset O_i$. There is a corresponding sequence of closed sets $S_i\equiv \Pi_A^{\oPP}-O_i$ and non-decreasing sequence of maximal values $m_i<1$ of $P_g(\rho_{AB}^{(1)}(\tau), \{ P_{A,i}\}_i)$ on $S_i$.
For each $m_i$ there is an $\epsilon_i$ such that for all $\epsilon<\epsilon_i$ the optimal $\oPP$-POVMs, namely the $\oPP$-POVMs that maximize $P_g(\rho_{AB}^{(1-\epsilon)}(\tau), \{P_{A,i}\}_i )$, belong to $O_i$. The sequence of $\epsilon_i$ is non-increasing since the sequence of $m_i$ is non-decreasing.

Let us consider a distance measure $d(\cdot,\cdot)$ on $\mathcal{B}(\mathcal{H}_{A})$ and define a sequence $O(\delta_i)$ of semi-open sets as the $\oPP$-POVMs $\{P_{A,i}\}_i$ such that $d( P_{A,i} , \kbra{i}_{A} )< \delta_i$ for any $i=1,...,\overline n$, for a strictly decreasing sequence $\delta_{i+1}<\delta_i$ where $\delta_i\to 0$ as $i\to \infty$. 

Then from the above argument we can conclude that, for any $\delta>0$ there exists a value $\lambda_\delta \in (0,1)$ such that, if $\lambda\in (\lambda_\delta ,1)$, any optimal {$\oPP$-POVM} $\{P_{A,i}^{(\lambda)} \}_i$ for this $\lambda$ is such that $d( P_{A,i}^{(\lambda)} , \kbra{i}_{A} )< \delta$ for any $i=1,...,\overline n$.

Next we show that $P_g(\rho_{AB}, \{ P_{A,i}\}_i  )$ is Lipschitz continuous as a function of $\{ P_{A,i}\}_i$. In other words, we construct a bound on the change of the guessing probability proportional to a distance measure quantifying the change of the POVM $\{ P_{A,i}\}_i$, valid for any $\rho_{AB}\in\mathcal{S}(\mathcal{H}_A\otimes\mathcal{H}_B)$.
We select a pair $\{ P_{A,i}^1\}_i$, $\{ P_{A,i}^2\}_i$ and observe that

\begin{eqnarray}\label{inin}
\max_{\{ P_{B,i} \}_i } \sum_i\mbox{Tr}[ P_{A,i}^1\otimes P_{B,i} \rho_{AB}]&=&\max_{ \{ P_{B,i} \}_i } \sum_i\mbox{Tr}[ P_{A,i}^2\otimes P_{B,i} \rho_{AB}+(P_{A,i}^1-P_{A,i}^2)\otimes P_{B,i}\rho_{AB}]\nonumber\\
&\leq &\max_{\{ P_{B,i} \}_i } \sum_i\mbox{Tr}[ P_{A,i}^2\otimes P_{B,i} \rho_{AB}]+\max_{ \{ P_{B,i} \}_i } \sum_i\mbox{Tr}[(P_{A,i}^1-P_{A,i}^2)\otimes P_{B,i}\rho_{AB}].
\end{eqnarray}

Let $\Delta_i$ be a diagonal matrix such that $\Delta_i= U_i(P_{A,i}^1-P_{A,i}^2)U^\dagger_i$ for a unitary $U_i$. Let $\Delta_{i+}$ and $\Delta_{i-}$ be the two diagonal positive semidefinite matrices such that $\Delta_i=\Delta_{i+}-\Delta_{i-}$. Note that $U^\dagger_i\Delta_{i+}U_i$ and $U^\dagger_i\Delta_{i-}U_i$ are positive semidefinite. This implies

\begin{eqnarray}
\max_{ \{ P_{B,i} \}_i } \sum_i \mbox{Tr}[(P_{A,i}^1-P_{A,i}^2)\otimes P_{B,i}\rho_{AB}]&=&
\max_{ \{ P_{B,i} \}_i }\sum_i \mbox{Tr}[ U^\dagger_i(\Delta_{i+}-\Delta_{i-})U_i\otimes P_{B,i}\rho_{AB}]\nonumber\\
&\leq &\max_{ \{ P_{B,i} \}_i }\sum_i \mbox{Tr}[ U^\dagger_i(\Delta_{i+})U_i\otimes P_{B,i}\rho_{AB}]+\max_{ \{ P_{B,i} \}_i }\sum_i \mbox{Tr}[ U^\dagger_i(\Delta_{i-})U_i\otimes P_{B,i}\rho_{AB}].
\end{eqnarray}
Since POVM elements are positive semidefinite $\mbox{Tr}[U^\dagger_i(\Delta_{i+})U_i\otimes P_{B,i}\rho_{AB}]$ is positive for each $P_{B,j}$.
Therefore $\mbox{Tr}[ U^\dagger_i(\Delta_{i+})U_i\otimes P_{B,i}\rho_{AB}]\leq\mbox{Tr}[U^\dagger_i(\Delta_{i+})U_i\otimes \sum_jP_{B,j}\rho_{AB}]=\mbox{Tr}[U^\dagger_i(\Delta_{i+})U_i\otimes \mathbbm{1}_B\rho_{AB}]$. Likewise $\mbox{Tr}[ U^\dagger_i(\Delta_{i-})U_i\otimes P_{B,i}\rho_{AB}]\leq \mbox{Tr}[U^\dagger_i(\Delta_{i-})U_i\otimes \mathbbm{1}_B\rho_{AB}]$. Using this we find that
\begin{eqnarray}\label{unin}
\max_{ \{ P_{B,i} \}_i } \sum_i \mbox{Tr}[(P_{A,i}^1-P_{A,i}^2)\otimes P_{B,i}\rho_{AB}]&\leq & \sum_i\mbox{Tr}[U^\dagger_i(\Delta_{i+}+\Delta_{i-})U_i\otimes \mathbbm{1}_B\rho_{AB}]\nonumber\\
&\leq &\sum_i\mbox{Tr}[U^\dagger_i(\Delta_{i+}+\Delta_{i-})U_i\otimes \mathbbm{1}_B]= (\overline{n}+1)d_S^2\sum_i\mbox{Tr}[\Delta_{i+}+\Delta_{i-}]=(\overline{n}+1)d_S^2\sum_i ||P_{A,i}^1-P_{A,i}^2||_1.\nonumber\\
\end{eqnarray}
where we used that $\mbox{Tr}[\mathbbm{1}_B]=(\overline{n}+1)d_S^2$ and for the second inequality we have used Von Neumann's trace inequality and that the largest eigenvalue of $\rho_{AB}$ is smaller or equal to 1. By combining Eq. (\ref{inin}) and Eq. (\ref{unin}) we can now conclude that
\begin{equation}
P_g(\rho_{AB}, \{ P_{A,i}^1\}_i  )-P_g(\rho_{AB}, \{ P_{A,i}^2\}_i  )\leq (\overline{n}+1)d_S^2\sum_i||P_{A,i}^1-P_{A,i}^2||_1.
\end{equation}
By exchanging the $\{ P_{A,i}^1\}_i$ and $\{ P_{A,i}^2\}_i$ in the above derivation we obtain 
\begin{equation}
P_g(\rho_{AB}, \{ P_{A,i}^2\}_i  )-P_g(\rho_{AB}, \{ P_{A,i}^1\}_i  )\leq (\overline{n}+1)d_S^2\sum_i||P_{A,i}^1-P_{A,i}^2||_1.
\end{equation}
Therefore
\begin{equation}\label{gnett}
|P_g(\rho_{AB}, \{ P_{A,i}^1\}_i  )-P_g(\rho_{AB}, \{ P_{A,i}^2\}_i  )|\leq (\overline{n}+1)d_S^2\sum_i||P_{A,i}^1-P_{A,i}^2||_1.
\end{equation}
Thus we have shown that $P_g(\rho_{AB}, \{ P_{A,i}\}_i  )$ is Lipschitz continuous as a function of $\{ P_{A,i}\}_i$ for any $\rho_{AB}\in\mathcal{S}(\mathcal{H}_A\otimes\mathcal{H}_B)$.

We now study the guessing probability of the ensemble that we obtain applying $\{P_{A,i}\}_i\in \Pi^\oPP_A$ on $\lrho(t)$ given by

\begin{equation}\label{genn}
P_g \left(\lrho(t) , \{P_{A,i} \}_i \right) \! = \! \lambda P_g \left( \{ \overline{p}_i , \rho_{A'',i}^{\perp }  \}_i \right) + (1-\lambda)  P_g \left( \{ \overline{p}_i , \rho_{SA',i}^{\parallel}(t)  \}_i \right) .
\end{equation}
We consider Eq. (\ref{genn}) when an optimal $\{ P_{A,i}^{(\lambda)} \}_i$ is chosen. We define the corresponding ensembles that appear in this expression  $\mathcal{E}^{\perp} (\{ P_{A,i}^{(\lambda)} \}_i)\equiv\{ \overline{p}_i , \rho_{A'',i}^{\perp }  \}_i$ and  $\mathcal{E}^{\parallel} (\{ P_{A,i}^{(\lambda)} \}_i)\equiv \{ \overline{p}_i , \rho_{SA',i}^{\parallel}(t)  \}_i$, so that 
\begin{equation}\label{optimm}
P_g \left(\lrho ( \tau ), \{ P_{A,i}^{(\lambda)} \}_i \right) = \lambda P_g (\mathcal{E}^{\perp} (\{ P_{A,i}^{(\lambda)} \}_i)) + (1-\lambda) P_g( \mathcal{E}^{\parallel} (\{ P_{A,i}^{(\lambda)} \}_i)) \, .
\end{equation}
The ensembles $\mathcal{E}^{\perp} (\{ P_{A,i}^{(\lambda)} \}_i)$ and $\mathcal{E}^{\parallel} (\{ P_{A,i}^{(\lambda)} \}_i)$ are functions on the set of optimal $\oPP$-POVMs $\{ P_{A,i}^{(\lambda)} \}_i$ for a given $\lambda$. Thus the image of the function $P_g (\mathcal{E}^{\perp} (\{ P_{A,i}^{(\lambda)} \}_i) )$ over the set of optimal $\oPP$-POVMs $\{ P_{A,i}^{(\lambda)} \}_i$ for a given $\lambda$, denoted $\mathrm{Im}(P_g^{(\lambda)} (\mathcal{E}^{\perp} ))\equiv\{P_g (\mathcal{E}^{\perp} (\{ P_{A,i}^{(\lambda)} \}_i) ):\{ P_{A,i}^{(\lambda)} \}_i \mathrm{\phantom{o} is\phantom{o}optimal}\}$, is a subset of the interval $[0,1]$, i.e., $\mathrm{Im}(P_g^{(\lambda)} (\mathcal{E}^{\perp} ))\subset[0,1]$. Likewise, the function $ P_g (\mathcal{E}^{\parallel}(\{ P_{A,i}^{(\lambda)} \}_i) )$ takes values in a set $\mathrm{Im}(P_g^{(\lambda)} (\mathcal{E}^{\parallel} ))\subset[0,1]$ for a given $\lambda$.

Using Eq. (\ref{gnett}) we can now construct bounds on $\mathrm{Im}(P_g^{(\lambda)} (\mathcal{E}^{\perp} ))$ and $\mathrm{Im}(P_g^{(\lambda)} (\mathcal{E}^{\parallel} ))$ for a given $\lambda$. First, based on the above argument we make the following observation: for any $\eta>0$ there exists a value $\lambda_\eta \in (0,1)$ such that, if $\lambda\in (\lambda_\eta ,1)$, any optimal {$\oPP$-POVM} $\{P_{A,i}^{(\lambda)} \}_i$ for this $\lambda$ is such that $|| P_{A,i}^{(\lambda)} -\kbra{i}_{A} ||_1< \eta$ for any $i=1,...,\overline n$. 
Thus, by Eq. (\ref{gnett}) the values in the image of $ P_g (\mathcal{E}^{\perp}(\{ P_{A,i}^{(\lambda)} \}_i) )$ for $\lambda\in (\lambda_\eta ,1)$ differ from $P_g (\mathcal{E}^{\perp}(\{ \kbra{i}_{A}\}_i) )=1$ by less than $\overline{n}(\overline{n}+1)d_S^2\eta$, i.e., $|P_g (\mathcal{E}^{\perp} (\{ P_{A,i}^{(\lambda)} \}_i) )-1|<\overline{n}(\overline{n}+1)d_S^2\eta$ for all optimal $\{ P_{A,i}^{(\lambda)} \}_i: \lambda \in (\lambda_\eta,1) $ . Likewise, the values in the range of $ P_g (\mathcal{E}^{\parallel} (\{ P_{A,i}^{(\lambda)} \}_i) )$ for $\lambda\in (\lambda_\eta ,1)$ differ from $ P_g (\mathcal{E}^{\parallel} (\{ \kbra{i}_{A}\}_i) )=P_g(\overline{\mathcal{E}}_{SA'}(\tau))$ by less than $\overline{n}(\overline{n}+1)d_S^2\eta$, i.e., $|P_g (\mathcal{E}^{\parallel}(\{ P_{A,i}^{(\lambda)} \}_i) )-P_g(\overline{\mathcal{E}}_{SA'}(\tau))|<\overline{n}(\overline{n}+1)d_S^2\eta$ for all optimal $\{ P_{A,i}^{(\lambda)} \}_i: \lambda \in (\lambda_\eta,1) $.
Using this we can state the following

\begin{equation}
\forall \delta >0,\, \exists \lambda_\delta >0 :  \, P_g (\mathcal{E}^{\parallel} (\{ P_{A,i}^{(\lambda)} \}_i)) - P_g(\overline{\mathcal{E}}_{SA'}(\tau)) < \delta \,\, , \, \forall \{ P_{A,i}^{(\lambda)} \}_i: \lambda \in (\lambda_\delta,1) \, .
\end{equation}

\section{Lipschitz continuity of $C^\PP_A$ on the set of states}\label{conti}

Consider a POVM $\{P_{A,i}\}_{i}$ and two states ${\rho}_{AB}$ and $\tilde{\rho}_{AB}$. Let $p_i=\mbox{Tr}[P_{A,i}{\rho}_{AB}]$ and $\tilde p_i=\mbox{Tr}[P_{A,i}\tilde{\rho}_{AB}]$.  
Let $\Delta$ be a diagonal matrix such that $\Delta= U(\tilde{\rho}_{AB}-\rho_{AB})U^\dagger$ for a unitary $U$. Let $\Delta_+$ and $\Delta_-$ be the two diagonal positive semidefinite matrices such that $\Delta=\Delta_+-\Delta_-$. Note that $U^\dagger\Delta_+U$ and $U^\dagger\Delta_-U$ are positive semidefinite. Then

\begin{eqnarray}\label{contin}
\tilde p_i-p_i=\mbox{Tr}[P_{A,i}(\tilde{\rho}_{AB}-{\rho}_{AB})]=\mbox{Tr}[P_{A,i}(U^\dagger\Delta_+U-U^\dagger\Delta_-U)]
\leq \mbox{Tr}[P_{A,i}U^\dagger\Delta_+U]+\mbox{Tr}[P_{A,i}U^\dagger\Delta_-U]
\end{eqnarray}
Since POVM elements are positive semidefinite $\mbox{Tr}[P_{A,j} U^\dagger\Delta_+U]$ is positive for each $P_{A,j}$.
Therefore $\mbox{Tr}[P_{A,i}U^\dagger\Delta_+U]\leq\mbox{Tr}[\sum_j P_{A,j} U^\dagger\Delta_+U]=\mbox{Tr}[U^\dagger\Delta_+U]=\mbox{Tr}[\Delta_+]$. Likewise $\mbox{Tr}[P_{A,i}(U^\dagger\Delta_-U)]\leq \mbox{Tr}[\Delta_-]$. Thus,
\begin{eqnarray}\label{contin}
\mbox{Tr}[P_{A,i}U^\dagger\Delta_+U]+\mbox{Tr}[P_{A,i}U^\dagger\Delta_-U]\leq \mbox{Tr}[\Delta_++\Delta_-]=||\tilde{\rho}_{AB}-{\rho}_{AB}||_1.
\end{eqnarray}
It follows that 
\begin{eqnarray}
\tilde p_i-p_i\leq||\tilde{\rho}_{AB}-{\rho}_{AB}||_1
\end{eqnarray}
By exchanging $p_i$ and $\tilde p_i$ in the above derivation we obtain
\begin{eqnarray}
p_i-\tilde p_i\leq||\tilde{\rho}_{AB}-{\rho}_{AB}||_1.
\end{eqnarray}
From this we can conclude that
\begin{eqnarray}\label{ecco}
| \tilde p_i- p_i|\leq||\tilde{\rho}_{AB}-{\rho}_{AB}||_1.
\end{eqnarray}

Assume now that $\{P_{A,i}\}_{i}$ is a $\PP$-POVM for ${\rho}_{AB}$ but not necessarily for $\tilde{\rho}_{AB}$.
We can create a $\PP$-POVM for $\tilde{\rho}_{AB}$ from $\{P_{A,i}\}_{i}$ in the following way. 
If $\tilde p_i- p_i>0$ we subtract $(1- p_i/\tilde p_i)P_{A,i}$ from $P_{A,i}$ to create a new element $\tilde P_{A,i}\equiv p_i/\tilde p_i P_{A,i}$. 
Let $P_{r}\equiv \sum_{i\in\{i+\}} (1- p_i/\tilde p_i)P_{A,i}$ where the $\{i+\}$ is the set of all $i$ such that $\tilde p_i- p_i>0$ and let $p_r\equiv\mbox{Tr}[P_{r}\tilde{\rho}_{AB}]=\sum_{i\in\{i+\}}\tilde p_i- p_i$.
If $\tilde p_i- p_i<0$ we add $( p_i-\tilde p_i)/(p_r)P_{r}$ to $P_{A,i}$ to create a new element $\tilde P_{A,i}\equiv P_{A,i}+( p_i-\tilde p_i)/(p_r)P_{r}$.

Next consider the trace distance between $\{ \tilde P_{A,i}\}_{i}$ and $\{P_{A,i}\}_{i}$.

\begin{eqnarray}
\sum_i||\tilde P_{A,i}-P_{A,i}||_1=\sum_{i\in\{i+\}}\left| \frac{\tilde p_i- p_i}{\tilde p_i}\right|||P_{A,i}||_1+\sum_{i\notin \{i+\} }\left|\frac{ p_i-\tilde p_i}{p_r }\right|||P_{r}||_1=\sum_{i\in\{i+\}}\left| \frac{\tilde p_i- p_i}{\tilde p_i}\right|||P_{A,i}||_1+||P_{r}||_1,
\end{eqnarray}
where we used that  $\sum_{i\notin \{i+\} }p_i-\tilde p_i=p_r$.
Since each $P_{A,i}$ is positive semidefinite with all eigenvalues less or equal to 1 it follows that 
$||P_{A,i}||_1\leq n_A$ where $n_A\equiv\textrm{dim}(\mathcal{H}_A)$. Moreover, $||P_{r}||_1=||\sum_{i\in\{i+\}} (1- p_i/\tilde p_i)P_{A,i}||_1\leq \sum_{i\in\{i+\}} |1- p_i/\tilde p_i|||P_{A,i}||_1$. Therefore,
\begin{eqnarray}
\sum_i||\tilde P_{A,i}-P_{A,i}||_1\leq 2\sum_{i\in\{i+\}}\left| \frac{\tilde p_i- p_i}{\tilde p_i}\right|||P_{A,i}||_1\leq 2n_A\sum_{i\in\{i+\}}\left| \frac{\tilde p_i- p_i}{\tilde p_i}\right|.
\end{eqnarray}
We further note that $\tilde p_i>p_i$ for $i\in\{i+\}$ and thus if $p_{min}\equiv \min_i p_i$ we have that $\tilde p_i>p_{min}$    for $i\in\{i+\}$. It follows that $|(\tilde p_i- p_i)/\tilde p_i|< |(\tilde p_i- p_i)/p_{min}|$ for $i\in\{i+\}$. Hence,

\begin{eqnarray}
\sum_i||\tilde P_{A,i}-P_{A,i}||_1 < \frac{2n_A}{p_{min}}\sum_{i\in\{i+\}}\left|\tilde p_i- p_i\right|\leq \frac{2n_A}{p_{min}}\sum_{i\in\{i+\}}||\tilde{\rho}_{AB}-{\rho}_{AB}||_1 < \frac{2n_A | \PP|}{p_{min}}||\tilde{\rho}_{AB}-{\rho}_{AB}||_1,
\end{eqnarray}
where $| \PP|$ is the number of elements of $\PP$ and we have used Eq. (\ref{ecco}).
Thus if $\{P_{A,i}\}_{i}$ is a $\PP$-POVM for ${\rho}_{AB}$ the minimum trace distance between $\{P_{A,i}\}_{i}$ and a $\PP$-POVM for $\tilde{\rho}_{AB}$ is upper bounded by $2n_A| \PP|||\tilde{\rho}_{AB}-{\rho}_{AB}||_1/p_{min} $. By an analogous argument if $\{\tilde P_{A,i}\}_{i}$ is a $\PP$-POVM for $\tilde{\rho}_{AB}$ the minimum trace distance between $\{\tilde P_{A,i}\}_{i}$ and a $\PP$-POVM for ${\rho}_{AB}$ is upper bounded by $2n_A| \PP|||\tilde{\rho}_{AB}-{\rho}_{AB}||_1/p_{min} $

We now recall Eq. (\ref{lip1}) and Eq. (\ref{gnett}) from Appendix \ref{appD} showing that the guessing probability $P_g(\rho_{AB}, \{ P_{A,i}\}_i  )$ is Lipschitz continuous  on the set of states for a fixed $ \{ P_{A,i}\}_i$ 

\begin{equation}\label{epple}
|P_g(\tilde\rho_{AB}, \{ P_{A,i}\}_i  )-P_g(\rho_{AB} \{ P_{A,i}\}_i  )|\leq ||\tilde{\rho}_{AB}-{\rho}_{AB}||_1,
\end{equation}
 and Lipschitz continuous on the set of POVMs for a fixed $\rho_{AB}$
\begin{equation}\label{gnettel}
|P_g(\rho_{AB}, \{ \tilde P_{A,i}\}_i  )-P_g(\rho_{AB}, \{ P_{A,i}\}_i  )|\leq n_B\sum_i||\tilde P_{A,i}-P_{A,i}||_1,
\end{equation}
where $n_B\equiv \textrm{dim}(\mathcal{H}_B)$.

We are now ready to show Lipschitz continuity of $C^\PP_A$ on the set of states. When $\rho_{AB}$ changes to $\tilde{\rho}_{AB}$ the minimum trace distance between any $\PP$-POVM for $\tilde{\rho}_{AB}$ and a $\PP$-POVM for ${\rho}_{AB}$ is upper bounded by $2n_A| \PP|||\tilde{\rho}_{AB}-{\rho}_{AB}||_1/p_{min} $. From this and Eq. (\ref{gnettel}) follows that the difference between the maximum of $P_g(\rho_{AB}, \{ P_{A,i}\}_i  )$ evaluated on the set $\Pi^\PP_A(\tilde\rho_{AB})$ of $\PP$-POVMs for $\tilde{\rho}_{AB}$ and the maximum of $P_g(\rho_{AB}, \{ P_{A,i}\}_i  )$ evaluated on the set $\Pi^\PP_A(\rho_{AB})$ of $\PP$-POVMs for ${\rho}_{AB}$
is upper bounded by $2n_An_B| \PP|||\tilde{\rho}_{AB}-{\rho}_{AB}||_1/p_{min}$. Moreover, by Eq. (\ref{epple}) the difference between $P_g(\rho_{AB}, \{ P_{A,i}\}_i  )$ and $P_g(\tilde{\rho}_{AB}, \{ P_{A,i}\}_i  )$ for any given $\{ P_{A,i}\}_i$ in the union $\Pi^\PP_A(\rho_{AB})\cup\Pi^\PP_A(\tilde\rho_{AB})$ of the set of $\PP$-POVMs for $\tilde{\rho}_{AB}$ and the set of $\PP$-POVMs for ${\rho}_{AB}$ is upper bounded by $||\tilde{\rho}_{AB}-{\rho}_{AB}||_1$. In conclusion the change of $C^\PP_A$ when $\rho_{AB}$ changes to $\tilde{\rho}_{AB}$ is upper bounded by $(1+2n_An_B| \PP|/p_{min})||\tilde{\rho}_{AB}-{\rho}_{AB}||_1$ , i.e.,

\begin{equation}\label{tettor}
|C^\PP_A(\tilde\rho_{AB} )-C^\PP_A(\rho_{AB})|< \left(1+\frac{2n_An_B| \PP|}{p_{min}}\right)||\tilde{\rho}_{AB}-{\rho}_{AB}||_1,
\end{equation}
Thus $C^\PP_A$ is Lipschitz continuous on the set of states. 

Using Eq. (\ref{tettor}) we can make some observations about the robustness of correlation backflows.
If we have a backflow in the interval $[\tau,\tau+\Delta\tau]$ for an initial state $\rho_{AB}(t_0)$, i.e., $C^\PP_A(\rho_{AB}(\tau+\Delta\tau) )-C^\PP_A(\rho_{AB}(\tau))>0$, any state $\rho'_{AB}$ such that $||{\rho'}_{AB}-{\rho}_{AB}(\tau+\Delta\tau)||_1< {p_{min}}/({p_{min}+2n_An_B| \PP|})|C^\PP_A(\rho_{AB} (\tau+\Delta\tau))-C^\PP_A(\rho_{AB}(\tau))|$ satisfies $C^\PP_A(\rho'_{AB})-C^\PP_A(\rho_{AB}(\tau))>0$. Likewise, if $C^\PP_A(\rho_{AB}(\tau+\Delta\tau) )-C^\PP_A(\rho_{AB}(\tau))>0$ any state $\rho''_{AB}$ such that $||{\rho''}_{AB}-{\rho}_{AB}(\tau)||_1< {p_{min}}/({p_{min}+2n_An_B| \PP|})|C^\PP_A(\rho_{AB} (\tau+\Delta\tau))-C^\PP_A(\rho_{AB}(\tau))|$ satisfies $C^\PP_A(\rho_{AB}(\tau+\Delta\tau))-C^\PP_A(\rho''_{AB})>0$. Moreover, if $C^\PP_A(\rho_{AB}(\tau+\Delta\tau) )-C^\PP_A(\rho_{AB}(\tau))>0$ any pair of states $\rho'_{AB}$ and $\rho''_{AB}$ such that $||{\rho'}_{AB}-{\rho}_{AB}(\tau+\Delta\tau)||_1+||{\rho''}_{AB}-{\rho}_{AB}(\tau)||_1< {p_{min}}/({p_{min}+2n_An_B| \PP|})|C^\PP_A(\rho_{AB} (\tau+\Delta\tau))-C^\PP_A(\rho_{AB}(\tau))|$ satisfies $C^\PP_A(\rho'_{AB})-C^\PP_A(\rho''_{AB})>0$.

Thus a backflow can be seen also for evolution of a perturbed initial state $\rho_{AB}(t_0)+\chi$ where $\chi$ is traceless Hermitian if $||\Lambda(\tau+\Delta\tau,t_0)\otimes \mathbbm{1}_B(\chi)||_1+||\Lambda(\tau,t_0)\otimes \mathbbm{1}_B(\chi)||_1< {p_{min}}/({p_{min}+2n_An_B| \PP|})|C^\PP_A(\rho_{AB} (\tau+\Delta\tau))-C^\PP_A(\rho_{AB}(\tau))|$. Since
$\Lambda(t,t_0)$ is CPTP for every $t$ it holds that $||\Lambda(\tau+\Delta\tau,t_0)\otimes \mathbbm{1}_B(\chi)||\leq ||\chi||$ and  $||\Lambda(\tau,t_0)\otimes \mathbbm{1}_B(\chi)||\leq ||\chi||$. Thus there is a neighbourhood of $\rho_{AB}(t_0)$  such that all states in this neighbourhood show a backflow in the interval $[\tau,\tau+\Delta\tau]$ and it includes all states $\rho_{AB}(t_0)+\chi$ such that $2||\chi||_1< {p_{min}}/({p_{min}+2n_An_B| \PP|})|C^\PP_A(\rho_{AB} (\tau+\Delta\tau))-C^\PP_A(\rho_{AB}(\tau))|$. Hence, this neighbourhood has the same dimension as $\mathcal{S}(\mathcal{H}_A \otimes \mathcal{H}_B)$.

\section{Comments on the Non-Markovianity measure: the case of non-differentiable $C^\PP_A(\rho_{ASA'}(t))$}

Here we discuss the non-Markovianity measure introduced in Eq. (\ref{measure}) and how it can be extended to work for almost everywhere differentiable $C^\PP_A(\rho_{ASA'}(t))$. We also comment on how one may construct measures of non-Markovianity based on $C^\PP_A(\rho_{ASA'}(t))$ using finite differences.

First we consider the case where ${C}_A^{\PP}(\rho_{ASA'}(t))$ is differentiable.
Consider the non-Markovianity measure introduced in Eq. (\ref{measure}) and let $[t_1,t_2]$ be a closed time interval for which it holds that $\frac{d}{dt}{C}_A^{\PP}(\rho_{ASA'}(t))>0$. In  Eq. (\ref{measure}) the type of integration used is not specified, but if the Henstock-Kurzweil integral is used it holds that

\begin{equation}\label{kurzweil}
 \int_{t_1}^{t_2} \frac{d}{dt}{C}_A^{\PP}(\rho_{ASA'}(t)) dt={C}_A^{\PP}(\rho_{ASA'}(t_2))-{C}_A^{\PP}(\rho_{ASA'}(t_1))\, ,
\end{equation}
if ${C}_A^{\PP}(\rho_{ASA'}(t))$ is differentiable in $[t_1,t_2]$. If the Riemann or Lebesgue integral is used there would be the additional requirement that $\frac{d}{dt}{C}_A^{\PP}(\rho_{ASA'}(t))$ is Riemann or Lebesgue integrable, respectively.

Next we consider the case where ${C}_A^{\PP}(\rho_{ASA'}(t))$ is almost everywhere differentiable, i.e. ${C}_A^{\PP}(\rho_{ASA'}(t))$ is non-differentiable for at most a countable set of times $t_i$.
At the times where ${C}_A^{\PP}(\rho_{ASA'}(t))$ fails to be differentiable, it is either non-differentiable but continuous or has a discontinuity. Since ${C}_A^{\PP}(\rho_{ASA'}(t))$ is a continuous function on the set of states it has a discontinuity only if the evolution of $\rho_{ASA'}(t)$ is discontinuous.
To deal with these points of non-differentiability we can define a function $\frac{d}{dt}{C}_A^{\PP}(\rho_{ASA'}(t))^*$ that is equal to $\frac{d}{dt}{C}_A^{\PP}(\rho_{ASA'}(t))$ for all $t$ for which $C^\PP_A(\rho_{ASA'}(t))$ is differentiable, and is equal to zero otherwise. If we use the Henstock-Kurzweil integral in the definition of the measure $N^\PP(\{\Lambda_S(t,t_0)\}_t)$ it is insensitive to how we define $\frac{d}{dt}{C}_A^{\PP}(\rho_{ASA'}(t))^*$ in the countable set of $t_i$ where $C^\PP_A(\rho_{ASA'}(t))$ is not differentiable. Thus we can define the measure

\begin{equation}\label{generalization}
N^\PP(\{\Lambda_S(t,t_0)\}_t) \equiv \sup_{\rho_{ASA'}(t_0)} \int_{\frac{d}{dt}{C}_A^{\PP}(\rho_{ASA'}(t))^*>0} \frac{d}{dt}{C}_A^{\PP}(\rho_{ASA'}(t))^* dt +\sum_{t_i}\Delta_{+}(t_i)\, ,
\end{equation}
where $\Delta_{+}(t_i)$ is the value of a discontinuous increase of $C^\PP_A(\rho_{ASA'}(t))$ at a time $t_i$. This definition reduces to that of Eq. (\ref{measure}) when ${C}_A^{\PP}(\rho_{ASA'}(t))$ is differentiable.

For the case when ${C}_A^{\PP}(\rho_{ASA'}(t))$ is not almost everywhere differentiable the measure in Eq. (\ref{generalization}) is not well defined. In this case one can resort to finite difference methods to estimate the amount of non-Markovianity in a given interval. A simple measure of this kind is

\begin{equation}
N^\PP_{finite}(\{\Lambda_S(t,t_0)\}_t) \equiv \sup_{\rho_{ASA'}(t_0), t_i<t_f} \{ 0 , C^\PP_A(\rho_{ASA'}(t_f)) - C^\PP_A(\rho_{ASA'}(t_i)) \},
\end{equation}
where $t_i$ and $t_f$ and belong to the interval of interest. We know that if the evolution is non-Markovian there always exists  at least one $\mathcal{P}$, some ancillas $A$ and $A'$, an initial state $\rho_{ASA'}(t_0)$ and a pair of times $t_{i}$ and $t_f$ such that $C^\PP_A(\rho_{ASA'}(t_f)) - C^\PP_A(\rho_{ASA'}(t_i))>0$ (See Theorem \ref{theorem}). Therefore, $N^\PP_{finite} (\{\Lambda_S(t,t_0)\}_t)>0$ if and only if the evolution $\{\Lambda_S(t,t_0)\}_t$ is non-Markovian.

\end{widetext}
\end{appendix}

\end{document}